\documentclass[acmtkdd]{acmtrans2m}
\usepackage{color}
\usepackage{url}
\usepackage{verbatim} 
\usepackage{subfigure} 
\usepackage{multirow}
\usepackage{times}
\usepackage{algorithm}
\usepackage[noend]{algorithmic}
\usepackage{xspace}
\usepackage{graphicx}
\usepackage{epsfig}
\usepackage{amsmath}
\usepackage{amssymb}
\usepackage{amsfonts}

\newcommand{\nedge}{k}
\newcommand{\bt}{\mathbf{t}}

\newcommand{\hide}[1]{}
\newcommand{\xhdr}[1]{\vspace{1.7mm}\noindent{{\bf #1.}}}
\newcommand{\eg}{\emph{e.g.}}
\newcommand{\ie}{\emph{i.e.}}

\newtheorem{proposition}{Proposition}
\newtheorem{theorem}{Theorem}

\newtheorem{problem}{Problem}

\newcommand{\UnnumberedFootnote}[1]{{\def\thefootnote{}\footnote{#1}
\addtocounter{footnote}{-1}}}

\newcommand{\netinf}{{\textsc{NetInf}}\xspace}
\newcommand{\netrate}{{\textsc{NetRate}}\xspace}
\newcommand{\connie}{{\textsc{CoNNIe}}\xspace}

\DeclareMathOperator*{\argmax}{argmax}

\newcommand{\BibTeX}{{\rm B\kern-.05em{\sc i\kern-.025em b}\kern-.08em
    T\kern-.1667em\lower.7ex\hbox{E}\kern-.125emX}}

\title{Inferring Networks of Diffusion and Influence}

\author{MANUEL GOMEZ-RODRIGUEZ\\Stanford University and MPI for Intelligent Systems\\
JURE LESKOVEC\\Stanford University\\
ANDREAS KRAUSE\\California Institute of Technology}

\begin{abstract}
Information diffusion and virus propagation are fundamental processes taking
place in networks. While it is often possible to directly ob\-ser\-ve when
nodes become infected with a virus or adopt the information, observing
in\-di\-vi\-dual transmissions (\ie, who infects whom, or who influences whom)
is typically very difficult. Furthermore, in many applications, the underlying
network over which the diffusions and propagations spread is actually {\em
unobserved}. We tackle these challenges by developing a method for tracing
paths of diffusion and influence through networks and inferring the networks
over which contagions propagate. Given the times when nodes adopt pieces of
information or become infected, we identify the optimal network that best
explains the observed infection times. Since the optimization problem is
NP-hard to solve exactly, we develop an efficient approximation algorithm that
scales to large datasets and finds pro\-va\-bly near-optimal networks.

We demonstrate the effectiveness of our approach by tracing information
diffusion in a set of 170 million blogs and news articles over a one year
period to infer how information flows through the online media space. We find
that the diffusion network of news for the top 1,000 media sites and blogs tends to 
have a core-periphery structure with a small set of core media sites that diffuse 
information to the rest of the Web. These sites tend to have stable circles of influence 
with more general news media sites acting as connectors between them.

\end{abstract}

\category{H.2.8}{Database Management}{Database applications}[Data mining]
\terms{Algorithms, Experimentation}
\keywords{Networks of diffusion, Information cascades, Blogs, News media, Meme-tracking, Social networks}

\begin{document}

\maketitle

\UnnumberedFootnote{Preliminary version of this work appeared in proceedings of
the 16th ACM SIGKDD International Conference on Knowledge Discovery and Data
Mining (KDD '10), 2010. Algorithm implementation and the data are available at
\url{http://snap.stanford.edu/netinf/}}

\section{Introduction}
\label{sec:intro}
The dissemination of information, cascading behavior, diffusion and spreading
of ideas, innovation, information, influence, viruses and diseases are
ubiquitous in social and information networks.
%
Such processes play a fundamental role in settings that include the spread of
technological innovations~\cite{rogers95diffusion,strang98diffusion}, word of
mouth effects in
marketing~\cite{domingos01mining,kempe03maximizing,jure06viral}, the spread of
news and
opinions~\cite{adar04blogspace,gruhl2004information,jure07cascades,leskovec2009kdd,nowell08letter},
collective problem-solving~\cite{kearns06experimental}, the spread of
infectious
diseases~\cite{anderson92infectious,bailey75mathematical,hethcote00diseases}
and sampling methods for hidden
populations~\cite{goodman61sampling,heckathorn97sampling}.

In order to study network diffusion there are two fundamental challenges one
has to address. First, to be able to track cascading processes taking place in
a network, one needs to identify the {\em contagion} (\ie, the idea,
information, virus, disease) that is actually spreading and propagating over
the edges of the network. Moreover, one has then to identify a way to
successfully trace the contagion as it is diffusing through the network. For
example, when tracing information diffusion, it is a non-trivial task to
automatically and on a large scale identify the phrases or ``memes'' that are
spreading through the Web~\cite{leskovec2009kdd}.

Second, we usually think of diffusion as a process that takes place on a {\em
network}, where the contagion propagates over the edges of the underlying
network from node to node like an epidemic. However, the network over which
propagations take place is usually {\em unknown} and {\em unobserved}.
Commonly, we only observe the times when particular nodes get ``infected'' but
we {\em do not} observe {\em who} infected them. In case of information
propagation, as bloggers discover new information, they write about it without
explicitly citing the source. Thus, we only observe the time when a blog gets
``infected'' with information but not where it got infected from. Similarly, in
virus propagation, we observe people getting sick without usually knowing who
infected them. And, in a viral marketing setting, we observe people purchasing
products or adopting particular behaviors without explicitly knowing who was
the influencer that caused the adoption or the purchase.

These challenges are especially pronounced in information diffusion on the Web,
where there have been relatively few large scale studies of information
propagation in large
networks~\cite{adar05epidemics,jure06influence,jure07cascades,nowell08letter}.
In order to study paths of diffusion over networks, one essentially requires to
have complete information about who influences whom, as a single missing link
in a sequence of propagations can lead to wrong
inferences~\cite{sadikov11cascades}. Even if one collects near complete large
scale diffusion data, it is a non-trivial task to identify textual fragments
that propagate relatively intact through the Web without human supervision. And
even then the question of how information diffuses through the network still
remains. Thus, the questions are, what is the network over which the information
propagates on the Web? What is the global structure of such a network? How do
news media sites and blogs interact? Which roles do different sites play in the
diffusion process and how influential are they?

\xhdr{Our approach to inferring networks of diffusion and influence}
We address the above questions by positing that there is some underlying
unknown network over which information, viruses or influence propagate. We
assume that the underlying network is static and does not change over time. We
then observe the times when nodes get infected by or decide to adopt a
particular contagion (a particular piece of information, product or a virus)
but we do not observe where they got infected from. Thus, for each contagion,
we only observe times when nodes got infected, and we are then interested in
determining the paths the diffusion took through the unobserved network. Our
goal is to reconstruct the network over which contagions propagate.
Figure~\ref{fig:FFN20} gives an example.

\begin{figure*}[h]
  \centering
  \subfigure[True network $G^*$]{\includegraphics[width=0.6\textwidth]{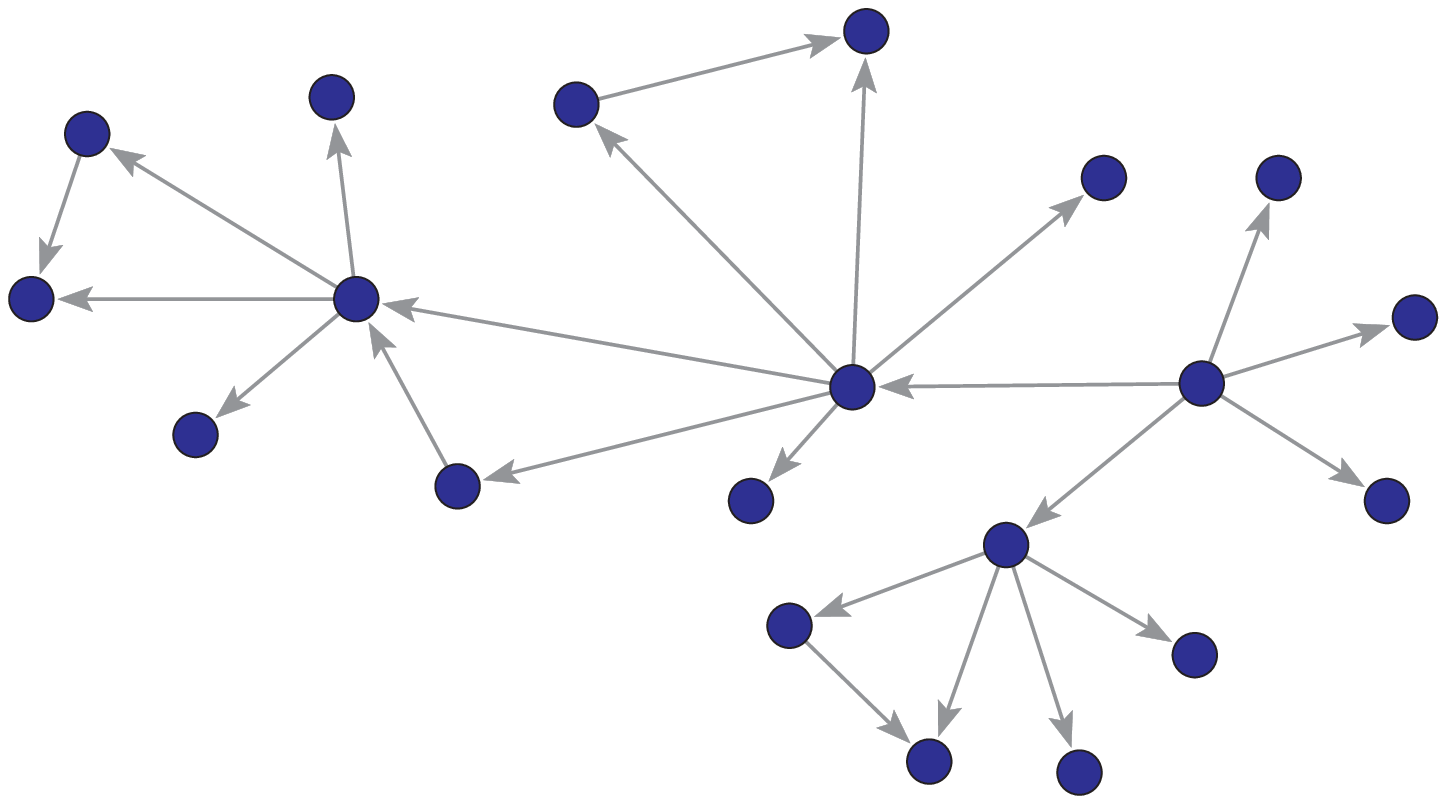}
    \label{fig:FFN20GroundTruth}} \\
  \subfigure[Inferred network $\hat{G}$ using heuristic baseline method]{\includegraphics[width=0.6\textwidth]{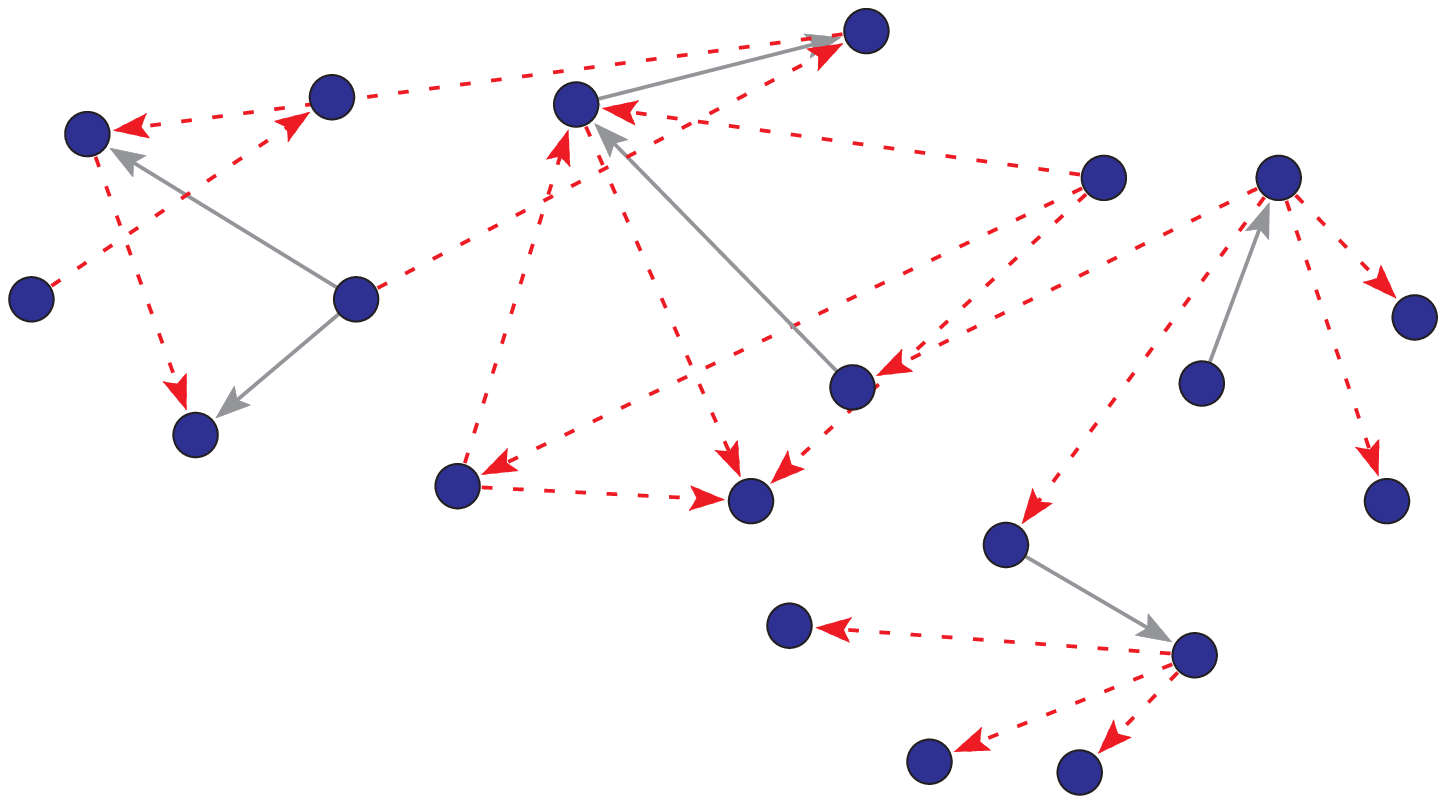}
    \label{fig:FFN20Baseline}} \\
  \subfigure[Inferred network $\hat{G}$ using \netinf algorithm]{\includegraphics[width=0.6\textwidth]{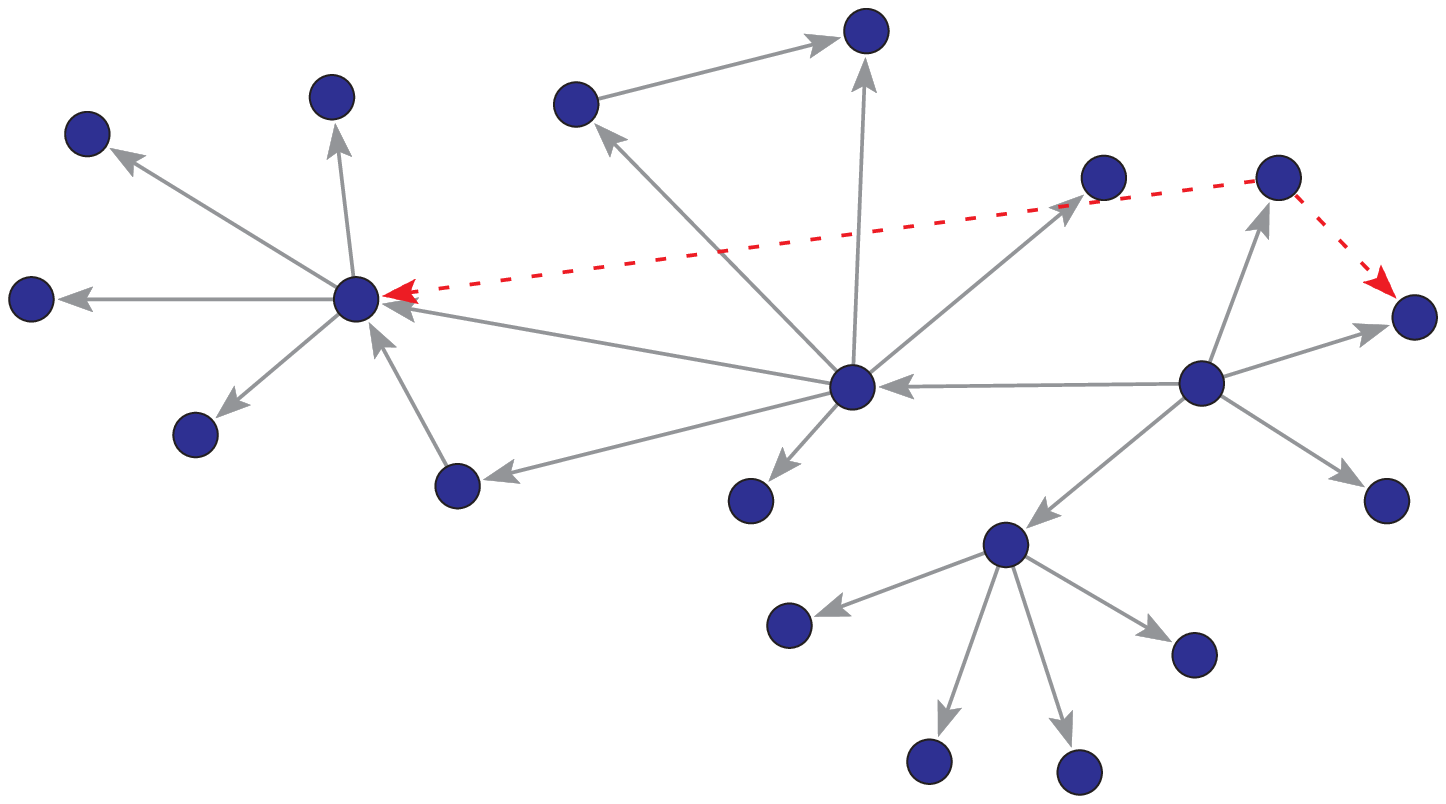}
    \label{fig:FFN20Estimated}}
  \caption{{\em Diffusion network inference problem.} There
  is an unknown network (a) over which contagions propagate. We are
  given a collection of node infection times and aim to recover the network in figure (a).
  Using a baseline heuristic (see Section~\ref{sec:evaluation}) we recover network
  (b) and using the proposed \netinf algorithm we recover network (c). Red edges denote mistakes. 
  The baseline makes many mistakes but \netinf almost perfectly recovers the network. }
  \label{fig:FFN20}
\end{figure*}

Edges in such networks of influence and diffusion have various interpretations.
In virus or disease propagation, edges can be interpreted as who-infects-whom.
In information propagation, edges are who-adopts-information-from-whom or
who-listens-to-whom. In a viral marketing setting, edges can be understood as
who-influences-whom.

The main premise of our work is that by observing many different contagions
spreading among the nodes, we can infer the edges of the underlying propagation
network. If node $v$ tends to get infected soon after node $u$ for many
different contagions, then we can expect an edge $(u,v)$ to be present in the
network. By exploring correlations in node infection times, we aim to recover
the unobserved diffusion network.

\begin{figure*}[t]
  \centering
  \includegraphics[width=0.7\textwidth]{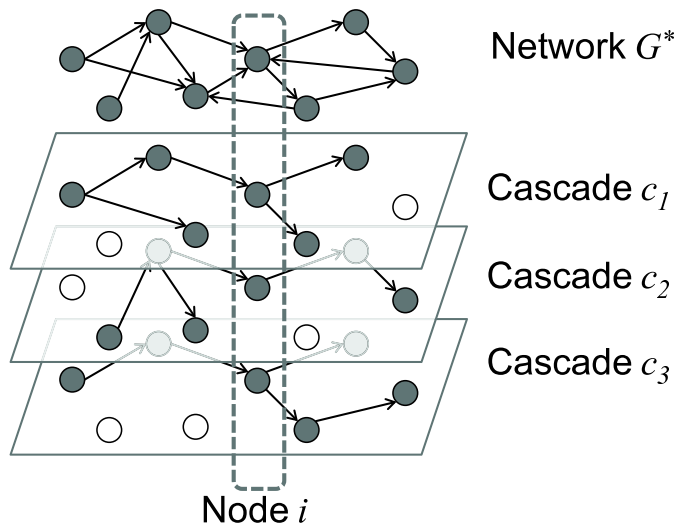}
\caption{The underlying true network over which contagions spread is
illustrated on the top. Each subsequent layer depicts a cascade created by the
diffusion of a particular contagion. For each cascade, nodes in gray are the
``infected'' nodes and the edges denote the direction in which the contagion
propagated. Now, given only the node infection times in each cascade we aim to
infer the connectivity of the underlying network $G^*$.}
  \label{fig:setcascades}
\end{figure*}

The concept of set of contagions over a network is illustrated in
Figure~\ref{fig:setcascades}. As a contagion spreads over the underlying
network it creates a trace, called {\em cascade}. Nodes of the cascade are the
nodes of the network that got infected by the contagion and edges of the cascade
represent edges of the network over which the contagion actually spread. On the
top of Figure~\ref{fig:setcascades}, the underlying true network over which
contagions spread is illustrated. Each subsequent layer depicts a cascade
created by a particular contagion. A priori, we do not know the connectivity of
the underlying true network and our aim is to infer this connectivity using the
infection times of nodes in many cascades.

We develop \netinf, a scalable algorithm for inferring networks of diffusion
and influence. We first formulate a generative probabilistic model of how, on a
fixed hypothetical network, contagions spread as directed trees (\ie, a
node infects many other nodes) through the network. Since we only observe times
when nodes get infected, there are many possible ways of the contagion could
have propagated through the network that are consistent with the observed data.
In order to infer the network we have to consider all possible ways of the
contagion spreading through the network. Thus, naive computation of the model
takes exponential time since there is a combinatorially large number of
propagation trees. We show that, perhaps surprisingly, computations over this
super-exponential set of trees can be performed in polynomial (cubic) time.
However, under such model, the network inference problem is still intractable.
Thus, we introduce a tractable approximation, and show that the objective
function can be both efficiently computed and efficiently optimized. By
exploiting a diminishing returns property of the problem, we prove that \netinf
infers near-optimal networks. We also speed-up \netinf by exploiting
the local structure of the objective function and by using lazy
evaluations~\cite{leskovec2007cost}.

In a broader context, our work here is related to the network structure
learning of probabilistic directed graphical models~\cite{friedman1999learning,getoor2003learning} 
where heuristic greedy hill-climbing or stochastic search that both offer 
no performance guarantees are usually used in practice. In contrast, our work 
here provides a novel formulation and a {\em tractable} polynomial time algorithm 
for inferring directed networks together with an approximation guarantee that 
ensures the inferred networks will be of near-optimal quality.

Our results on synthetic datasets show that we can reliably infer an
underlying propagation and influence network, regardless of the overall
network structure. Validation on real and synthetic datasets shows that \netinf
outperforms a baseline heuristic by an order of magnitude and correctly
discovers more than 90\% of the edges. We apply our algorithm to a real Web
information propagation dataset of 170 million blog and news articles over a
one year period.  Our results show that online news propagation networks tend
to have a core-periphery structure with a small set of core blog and news media
websites that diffuse information to the rest of the Web, news media websites
tend to diffuse the news faster than blogs and blogs keep discussing about news
longer time than media websites.

Inferring how information or viruses propagate over networks is crucial for a
better understanding of diffusion in networks. By modeling the structure of the
propagation network, we can gain insight into positions and roles various nodes
play in the diffusion process and assess the range of influence of nodes in the
network.

The rest of the paper is organized as follows. Section~\ref{sec:formulation} is
devoted to the statement of the problem, the formulation of the model and the
optimization problem. In section~\ref{sec:proposed}, an efficient reformulation
of the optimization problem is proposed and its solution is presented.
Experimental evaluation using synthetic and MemeTracker data are shown in
section~\ref{sec:evaluation}. We conclude with related work in
section~\ref{sec:related} and discussion of our results in
section~\ref{sec:conclusions}.

\section{Diffusion network inference problem}
\label{sec:formulation}
We next formally describe the problem where contagions propagate over an
unknown static directed network and create cascades. For each cascade we
observe times \emph{when} nodes got infected but not \emph{who} infected
them. The goal then is to infer the unknown network over which contagions 
originally propagated. In an information diffusion setting, each contagion 
corresponds to a different piece of information that spreads over the network 
and all we observe are the times when particular nodes adopted or mentioned 
particular information. The task then is to infer the network where a directed 
edge $(u,v)$ carries the semantics that node $v$ tends to get influenced by 
node $u$ (\ie, mentions or adopts the information after node $u$ does 
so as well).

\subsection{Problem statement}
\label{sec:problem}

Given a hidden directed network $G^*$, we observe multiple contagions spreading
over it. As the contagion $c$ propagates over the network, it leaves a trace, a
cascade, in the form of a set of triples $(u, v, t_v)_c$, which means that contagion 
$c$ reached node $v$ at time $t_v$ by spreading from node $u$ (\ie, by 
propagating over the edge $(u,v)$). We denote the fact that the cascade initially 
starts from some active node $v$ at time $t_{v}$ as $(\emptyset, v, t_v)_c$.

Now, we only get to observe the time $t_v$ when contagion $c$ reached node $v$
but not {\em how} it reached the node $v$, \ie, we only know that $v$
got infected by one of its neighbors in the network but do not know who $v$'s
neighbors are and who of them infected $v$. Thus, instead of observing the
triples $(u, v, t_v)_c$ that fully specify the trace of the contagion $c$
through the network, we only get to observe pairs $(v, t_v)_c$ that describe
the time $t_v$ when node $v$ got infected by the contagion $c$. Now, given 
such data about node infection times for many different contagions, we aim to
recover the unobserved directed network $G^*$, \ie, the network over which the
contagions originally spread.

We use the term {\em hit time} $t_u$ to refer to the time when a cascade created by a
contagion hits (infects, causes the adoption by) a particular node $u$. In
practice, many contagions do not hit all the nodes of the network. Simply, if a
contagion hits all the nodes this means it will infect every node of the network.
In real-life most cascades created by contagions are relatively small. Thus, if
a node $u$ is not hit by a cascade, then we set $t_u = \infty$. Then, the observed
data about a cascade $c$ is specified by the vector $\mathbf{t_c}=[t_{1},\dots,t_{n}]$ 
of hit times, where $n$ is the number of nodes in $G^{*}$, and $t_i$ is the time when 
node $i$ got infected by the contagion $c$ ($t_i=\infty$ if $i$ did not get infected 
by $c$).

Our goal now is to infer the network $G^*$. In order to solve this problem we
first define the probabilistic model of how contagions spread over the edges of
the network.
We first specify the contagion transmission model $P_c(u, v)$ that describes
how likely is that node $u$ spreads the contagion $c$ to node $v$. Based on the
model we then describe the probability $P(c | T)$ that the contagion $c$
propagated in a particular cascade tree pattern $T=(V_T, E_T)$, where tree $T$
simply specifies which nodes infected which other nodes (\eg, see
Figure~\ref{fig:setcascades}). Last, we define $P(c | G)$, which is the
probability that cascade $c$ occurs in a network $G$. And then, under this
model, we show how to estimate a (near-)maximum likelihood network $\hat{G}$, 
\ie, the network $\hat{G}$ that (approximately) maximizes the probability of 
cascades $c$ occurring in it.

\subsection{Cascade Transmission Model}

We start by formulating the probabilistic model of how contagions diffuse over
the network. We build on the Independent Cascade Model~\cite{kempe03maximizing}
which posits that an infected node infects each of its neighbors in the network
$G$ independently at random with some small chosen probability. This model
implicitly assumes that every node $v$ in the cascade $c$ is infected by at
most one node $u$. That is, it only matters when the first neighbor of $v$
infects it and all infections that come afterwards have no impact. Note that $v$ can 
have multiple of its neighbors infected but only one neighbor actually activates $v$.
Thus, the structure of a cascade created by the diffusion of contagion $c$ is fully
described by a directed tree $T$, that is contained in the directed graph $G$,
\ie, since the contagion can only spread over the edges of $G$ and each node
can only be infected by at most one other node, the pattern in which the
contagion propagated is a tree and a subgraph of $G$. 
Refer to Figure~\ref{fig:setcascades} for an illustration of a network and a set of
cascades created by contagions diffusing over it.

\xhdr{Probability of an individual transmission} The Independent Contagion
Model only implicitly models time through the epochs of the propagation. We
thus formulate a variant of the model that preserves the tree structure of
cascades and also incorporates the notion of time.

We think of our model of how a contagion transmits from $u$ to $v$ in two
steps. When a new node $u$ gets infected it gets a chance to transmit the
contagion to each of its currently uninfected neighbors $w$ independently with
some small probability $\beta$. If the contagion is transmitted we then sample
{\em the incubation time}, \ie, how long after $w$ got infected, $w$ will get a
chance to infect its (at that time uninfected) neighbors. Note that cascades in
this model are necessarily trees since node $u$ only gets to infect neighbors
$w$ that have not yet been infected.

First, we define the probability $P_{c}(u,v)$ that a node $u$ spreads the
cascade to a node $v$, \ie, a node $u$ influences/infects/transmits contagion
$c$ to a node $v$. Formally, $P_{c}(u,v)$ specifies the conditional probability
of observing cascade $c$ spreading from $u$ to $v$.

Consider a pair of nodes $u$ and $v$, connected by a directed edge $(u,v)$ and
the corresponding hit times $(u, t_u)_c$ and $(v, t_v)_c$. Since the contagion can
only propagate forward in time, if node $u$ got infected after node
$v$ ($t_u > t_v$) then $P_{c}(u,v)=0$, \ie, nodes can not influence nodes from
the past. On the other hand (if $t_u < t_v$) we make no assumptions about the
properties and shape of $P_{c}(u,v)$. To build some intuition, we can think
that the probability of propagation $P_{c}(u,v)$ between a pair of nodes $u$
and $v$ is decreasing in the difference of their infection times, \ie, the
farther apart in time the two nodes get infected the less likely they are to
infect one another.

However, we note that our approach allows for the contagion transmission model
$P_c(u,v)$ to be arbitrarily complicated as it can also depend on the
properties of the contagion $c$ as well as the properties of the nodes and
edges. For example, in a disease propagation scenario, node attributes could
include information about the individual's socio-economic status, commute
patterns, disease history and so on, and the contagion properties would 
include the strength and the type of the virus. This allows for great flexibility 
in the cascade transmission models as the probability of infection depends 
on the parameters of the disease and properties of the nodes.

\begin{table}
\centering
\begin{tabular}{l|l}
  {\bf Symbol} & {\bf Description} \\ \hline
  $G(V,E)$ & Directed graph with nodes $V$ and edges $E$ over which contagions spread\\
  $\beta$ & Probability that contagion propagates over an edge of $G$\\
  $\alpha$ & Incubation time model parameter (refer to Eq.~\ref{eq:incubation}) \\
  $E_\varepsilon$ & Set of $\varepsilon$-edges, $E\cap E_\varepsilon=\emptyset$ and $E \cup E_\varepsilon = V \times V$ \\
  $c$ & Contagion that spreads over $G$\\
  $t_u$ & Time when node $u$ got hit (infected) by a particular cascade \\
  $\bt_c$ & Set of node hit times in cascade $c$. $\bt_c[i]=\infty$ if node $i$ did not participate in $c$\\
  $\Delta_{u, v}$ & Time difference between the node hit times $t_v - t_u$ in a particular cascade \\
  $C = \{(c, \bt_c)\}$ & Set of contagions $c$ and corresponding hit times, \ie, the observed data\\
  $\mathcal{T}_c(G)$ & Set of all possible propagation trees of cascade $c$ on graph $G$ \\
  $T(V_T, E_T)$ & Cascade propagation tree, $T \in \mathcal{T}_c(G)$\\
  $V_T$ & Node set of $T$, $V_T =\{i\ |\ i \in V \text{and}\ \bt_c[i] < \infty \}$ \\
  $E_T$ & Edge set of $T$, $E_T \subseteq E \cup E_\varepsilon$ \\
\end{tabular}
\caption{Table of symbols.}
\end{table}

Purely for simplicity, in the rest of the paper we assume the simplest and most intuitive model 
where the probability of transmission depends only on the time difference between the node hit 
times $\Delta_{u, v} = t_v-t_u$. We consider two different models for the incubation time distribution 
$\Delta_{u, v}$, an exponential and a power-law model, each with parameter $\alpha$:

\begin{equation}
  P_c(u, v) = P_c(\Delta_{u, v}) \propto e^{-\frac{\Delta_{u, v}}{\alpha}} \mbox{  and  }
  P_c(u, v) = P_c(\Delta_{u, v}) \propto \frac{1}{\Delta_{u, v}^{\alpha}}.
  \label{eq:incubation}
\end{equation}

Both the power-law and exponential waiting time models have been argued for in the 
literature~\cite{barabasi05human,jure07cascades,malgrem08poisson}. In the end, our algorithm 
does not depend on the particular choice of the incubation time distribution and more complicated 
non-monotonic and multimodal functions can easily be chosen~\cite{crane08response,wallinga04epidemic,manuel11icml}. 
Also, we interpret $\infty+\Delta_{u, v}=\infty$, \ie, if $t_{u}=\infty$, then $t_{v}=\infty$ with probability 
$1$. Note that the parameter $\alpha$ can potentially be different for each edge in the network.

Considering the above model in a generative sense, we can think that the cascade $c$ reaches 
node $u$ at time $t_u$, and we now need to generate the time $t_v$ when $u$ spreads the cascade 
to node $v$. As cascades generally do not infect all the nodes of the network, we need to explicitly 
model the probability that the cascade stops. With probability $(1-\beta)$, the cascade stops, and 
never reaches $v$, thus $t_{v}=\infty$. With probability $\beta$, the cascade transmits over the edge 
$(u,v)$, and the hit time $t_{v}$ is set to $t_{u}+\Delta_{u, v}$, where $\Delta_{u, v}$ is the incubation 
time that passed between the hit times $t_{u}$ and $t_{v}$.

\xhdr{Likelihood of a cascade spreading in a given tree pattern $\bf T$} Next we calculate the likelihood 
$P(c|T)$ that contagion $c$ in a graph $G$ propagated in a particular tree pattern $T(V_T, E_T)$ under some
assumptions. This means we want to assess the probability that a cascade $c$ with hit times $\bt_c$ propagated 
in a particular tree pattern $T$.

Due to our modeling assumption that cascades are trees the likelihood is simply:

\begin{equation}
P(c|T) = \prod_{(u,v)\in E_T} \beta P_c(u, v) \prod_{u \in V_T, (u,x) \in E\backslash E_T} (1-\beta),
    \label{eq:ctree}
\end{equation}
where $E_T$ is the edge set and $V_T$ is the vertex set of tree $T$. Note that
$V_T$ is the set of nodes that got infected by $c$, \ie, $V_T \subset V$ and
contains elements $i$ of $\bt_c$ where $\bt_c(i)<\infty$. The above expression
has an intuitive explanation. Since the cascade spread in tree pattern $T$, the
contagion successfully propagated along those edges. And, along the edges where
the contagion did not spread, the cascade had to stop. Here, we assume independence 
between edges to simplify the problem. Despite this simplification, we later show empirically that 
\netinf works well in practice

Moreover, $P(c|T)$ can be rewritten as:

\begin{equation}
P(c|T) = \beta^{q} (1-\beta)^{r} \prod_{(u,v)\in E_T} P_c(u, v),
    \label{eq:ctree2}
\end{equation}
where $q = |E_T| = |V_T|-1$ is the number of edges in $T$ and counts the edges
over which the contagion successfully propagated. Similarly, $r$ counts the
number of edges that did not activate and failed to transmit the contagion: $r
= \sum_{u \in V_T} d_{out}(u) - q$, and $d_{out}(u)$ is the out-degree of node
$u$ in graph $G$.

We conclude with an observation that will come very handy later. Examining
Eq.~\ref{eq:ctree2} we notice that the first part of the equation before the
product sign does not depend on the edge set $E_T$ but only on the vertex set
$V_T$ of the tree $T$. This means that the first part is constant for all trees
$T$ with the same vertex set $V_T$ but possibly different edge sets $E_T$. For
example, this means that for a fixed $G$ and $c$ maximizing $P(c|T)$ with
respect to $T$ (\ie, finding the most probable tree), does not depend
on the second product of Eq.~\ref{eq:ctree}. This means that when optimizing,
one only needs to focus on the first product where the edges of the tree $T$
simply specify how the cascade spreads, \ie, every node in the cascade gets
influenced by exactly one node, that is, its parent.

\xhdr{Cascade likelihood} We just defined the likelihood $P(c|T)$ that a single
contagion $c$ propagates in a particular tree pattern $T$. Now, our aim is to
compute $P(c|G)$, the probability that a cascade $c$ occurs in a graph $G$.
Note that we observe only the node infection times while the exact propagation
tree $T$ (who-infected-whom) is unknown. In general, over a given graph $G$
there may be multiple different propagation trees $T$ that are consistent with
the observed data. For example, Figure~\ref{fig:mtrees} shows three different
cascade propagation paths (trees $T$) that are all consistent with the observed
data $\bt_c=(t_a=1, t_c=2, t_b=3, t_e=4)$

\begin{figure}[t] 
  \centering
  \includegraphics[width=\textwidth]{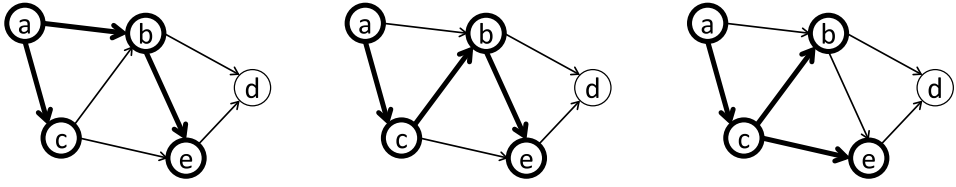}
\caption{Different propagation trees $T$ of cascade $c$ that are all
consistent with observed node hit times $c=(t_a=1, t_c=2, t_b=3, t_e=4)$.
In each case, wider edges compose the tree, while thinner edges are the rest of
the edges of the network $G$.}
\label{fig:mtrees}
\end{figure}

So, we need to combine the probabilities of individual propagation trees into a 
probability of a cascade $c$. We achieve this by considering all possible propagation 
trees $T$ that are supported by network $G$, \ie, all possible ways in which cascade 
$c$ could have spread over $G$:

\begin{equation}
  P(c | G) = \sum_{T \in \mathcal{T}_c(G)} P(c|T) P(T|G),
    \label{eq:pcasc}
\end{equation}
where $c$ is a cascade and $\mathcal{T}_c(G)$ is the set of all the directed
connected spanning trees on a subgraph of $G$ induced by the nodes that got 
hit by cascade $c$. Note that even though the sum ranges over all possible 
spanning trees $T \in \mathcal{T}_c(G)$, in case $T$ is inconsistent with the 
observed data, then $P(c|T)=0$.

Assuming that all trees are a priori equally likely (\ie, $P(T|G)=1/|\mathcal{T}_c(G)|$) 
and using the observation from Equation~\ref{eq:ctree2} we obtain:

\begin{equation}
  P(c | G) \propto \sum_{T \in \mathcal{T}_c(G)} \prod_{(u,v)\in E_T} P_c(u, v).
    \label{eq:pcasc2}
\end{equation}

Basically, the graph $G$ defines the skeleton over which the cascades can
propagate and $T$ defines a particular possible propagation tree. There may be
many possible trees that \emph{explain} a single cascade (see
Fig.~\ref{fig:mtrees}), and since we do not know in which particular tree
pattern the cascade really propagated, we need to consider all possible
propagation trees $T$ in $\mathcal{T}_c(G)$. Thus, the sum over $T$ is a sum
over all directed spanning trees of the graph induced by the vertices that got
hit by the cascade $c$.

We just computed the probability of a single cascade $c$ occurring in $G$, and we
now define the probability of a set of cascades $C$ occurring in $G$ simply as:

\begin{equation}
  P(C| G) = \prod_{c \in C} P(c | G),
  \label{eq:pgraph}
\end{equation}
where we again assume conditional independence between cascades given the graph
$G$.

\subsection{Estimating the network that maximizes the cascade likelihood}

Now that once we have formulated the cascade transmission model, we now 
state the {\em diffusion network inference problem}, where the goal is to 
find $\hat G$ that solves the following optimization problem:

\begin{problem}
Given a set of node infection times $\bt_c$ for a set of cascades $c \in C$,
a propagation probability parameter $\beta$ and an incubation time distribution
$P_{c}(u,v)$, find the network $\hat{G}$ such that:

\begin{equation}
  \hat G = \argmax_{|G|\leq\nedge} P(C|G),
  \label{eq:pgraph2}
\end{equation}
where the maximization is over all directed graphs $G$ of at most $\nedge$
edges, and $P(C|G)$ is defined by equations~\ref{eq:pgraph},~\ref{eq:pcasc}
and~\ref{eq:ctree}.
\end{problem}

We include the constraint on the number of edges in $\hat{G}$ simply because we seek
for a sparse solution, since real graphs are sparse. We discuss how to choose $\nedge$ 
in further sections of the paper.

The above optimization problem seems wildly intractable. To evaluate 
Eq.~\eqref{eq:pgraph}, we need to compute Eq.~\eqref{eq:pcasc} for each cascade
$c$, \ie, the sum over all possible spanning trees $T$. The number of trees can
be super-exponential in the size of $G$ but perhaps surprisingly, this
super-exponential sum can be performed in time polynomial in the number $n$ of
nodes in the graph $G$, by applying Kirchhoff's matrix tree
theorem~\cite{knuth1968art}:

\begin{theorem}[\cite{tutte48matrix}] \label{thm:tutte}
  If we construct a matrix $A$ such that
  $a_{i,j} = \sum_k w_{k, j}$ if $i = j$ and $a_{i,j} = -w_{i,j}$ if $i \neq j$
  and if $A_{x,y}$ is the matrix created by removing any row $x$ and column $y$ from $A$, then
  \begin{equation}\label{eq:matrixtree}
  (-1)^{x+y} \det(A_{x,y}) = \sum_{T \in A} \prod_{(i,j) \in T} w_{i,j},
  \end{equation}
  where $T$ is each directed spanning tree in $A$.
\end{theorem}

In our case, we set $w_{i,j}$ to be simply $P_c(i,j)$ and we compute the
product of the determinants of $|C|$ matrices, one for each cascade, which is
exactly Eq.~\ref{eq:pcasc}. Note that since edges $(i,j)$ where $t_{i}\geq t_{j}$ 
have weight 0 (i.e., they are not present), given a fixed cascade $c$, the collection 
of edges with po\-si\-tive weight forms a directed \emph{acyclic} graph (DAG). A 
DAG with a time-ordered labeling of its nodes has an upper triangular connectivity 
matrix. Thus, the matrix $A_{x, y}$ of Theorem~\ref{thm:tutte} is, by construction, upper 
triangular. Fortunately, the determinant of an upper triangular matrix is simply the product 
of the elements of its diagonal. This means that instead of using super-exponential time, we 
are now able to evaluate Eq.~\ref{eq:pgraph} in time $(|C| \cdot |V|^2)$ (the time required to 
build $A_{x, y}$ and compute the determinant for each of the $|C|$ cascades).

However, this does not completely solve our problem for two reasons: First,
while cuadratic time is a drastic improvement over a super-exponential computation,
it is still too expensive for the large graphs that we want to consider.
Second, we can use the above result only to evaluate the quality of a
\emph{particular} graph $G$, while our goal is to find the best graph
$\hat{G}$. To do this, we would need to search over \emph{all} graphs $G$ to
find the best one. Again, as there is a super-exponential number of graphs,
this is not practical. To circumvent this one could propose some ad hoc search
heuristics, like hill-climbing. However, due to the combinatorial nature of the
likelihood function, such a procedure would likely be prone to local maxima. We
leave the question of efficient maximization of Eq.~\ref{eq:pcasc} where
$P(c|G)$ considers all possible propagation trees as an interesting open
problem.

\section{Alternative formulation and the \netinf algorithm}
\label{sec:proposed}
The diffusion network inference problem defined in the previous section does
not seem to allow for an efficient solution. We now propose an alternative
formulation of the problem that is tractable both to compute and also to
optimize.

\subsection{An alternative formulation}

We use the same tree cascade formation model as in the previous section.
However, we compute an approximation of the likelihood of a single cascade by
considering only the most likely tree instead of all possible propagation trees. We 
show that this approximate likelihood is tractable to compute. Moreover, we
also devise an algorithm that provably finds networks with near optimal
approximate likelihood. In the remainder of this section, we informally write 
likelihood and log-likelihood even though they are approximations. However, all 
approximations are clearly indicated.

First we introduce the concept of $\varepsilon$-edges to account for the fact
that nodes may get infected for reasons other than the network influence. For
example, in online media, not all the information propagates via the network,
as some is also pushed onto the network by the mass
media~\cite{katz1955personal,dodds07influentials} and thus a disconnected
cascade can be created. Similarly, in viral marketing, a person may purchase a
product due to the influence of peers (\emph{i.e.}, network effect) or for some
other reason (\eg, seing a commercial on TV)~\cite{jure06viral}.

\xhdr{Modeling external influence via $\varepsilon$-edges} To account for such
phenomena when a cascade ``jumps'' across the network we can think of creating
an additional node $x$ that represents an {\em external influence} and can
infect \emph{any} other node $u$ with small probability. We then connect the
external influence node $x$ to every other node $u$ with an $\varepsilon$-edge.
And then every node $u$ can get infected by the external source $x$ with a very
small probability $\varepsilon$. For example, in case of information diffusion
in the blogosphere, such a node $x$ could model the effect of blogs getting infected
by the mainstream media.

However, if we were to adopt this approach and insert an additional external
influence node $x$ into our data, we would also need to infer the edges
pointing out of $x$, which would make our problem even harder. Thus, in order
to capture the effect of external influence, we introduce a concept of
$\varepsilon$-edge. If there is not a network edge between a node $i$ and a node $j$ 
in the network, we add an $\varepsilon$-edge and then node $i$ can infect node 
$j$ with a small probability $\varepsilon$. Even though adding $\varepsilon$-edges makes 
our graph $G$ a clique (\ie, the union of network edges and $\varepsilon$-edges creates a
clique), the $\varepsilon$-edges play the role of external influence node.

Thus, we now think of graph $G$ as a fully connected graph of two disjoint sets
of edges, the network edge set $E$ and the $\varepsilon$-edge set
$E_\varepsilon$, \ie, $E \cap E_\varepsilon = \emptyset$ and $E \cup
E_\varepsilon = V \times V$.

Now, any cascade propagation tree $T$ is a combination of network and
$\varepsilon$-edges. As we model the external influence via the
$\varepsilon$-edges, the probability of a cascade $c$ occurring in tree $T$
(\ie, the analog of Eq.~\ref{eq:ctree}) can now be computed as:

\begin{equation}
  P(c|T) = \prod_{u \in V_T} \ \ \prod_{v \in V} P_c'(u,v),
  \label{eq:probcasc2}
\end{equation}
where we compute the transmission probability $P_c'(u,v)$ as follows:

\begin{align*}
P_c'(u,v)=
\begin{cases}
  \beta P_c(t_v-t_u) & \text{if $t_u < t_v$ and $(u,v) \in E_T \cap E$ \ \ \ \ \ \ \ {\em $(u,v)$ is network edge}}\\
  \varepsilon P_c(t_v-t_u) & \text{if $t_u < t_v$ and $(u,v) \in E_T \cap E_\varepsilon$ \ \ \ \ \ {\em $(u,v)$ is $\varepsilon$-edge}}\\
  1-\beta        & \text{if $t_v=\infty$ and $(u,v) \in E \backslash E_T$ \ \ \ \ \ \ \ \ \ {\em $v$ is not infected, network edge}}\\
  1-\varepsilon  & \text{if $t_v=\infty$ and $(u,v) \in E_\varepsilon \backslash E_T$ \ \ \ \ \ \ \ {\em $v$ is not infected, $\varepsilon$-edge}}\\
  0              & \text{else (\ie, $t_u \ge t_v$).}\\
\end{cases}
\end{align*}

Note that above we distinguish four type of edges: network and
$\varepsilon$-edges that participated in the diffusion of the contagion and 
network and $\varepsilon$-edges that did not participate in the diffusion.

Figure~\ref{fig:epsedges} further illustrates this concept. First,
Figure~\ref{fig:epsedges1} shows an example of a graph $G$ on five nodes and
four network edges $E$ (solid lines), and any other possible edge is the
$\varepsilon$-edge (dashed lines). Then, Figure~\ref{fig:epsedges2} shows an
example of a propagation tree $T=\{(a,b), (b,c), (b,d)\}$ in graph $G$. We only
show the edges that play a role in Eq.~\ref{eq:probcasc2} and label them with
four different types: (a) network edges that transmitted the contagion, (b)
$\varepsilon$-edges that transmitted the contagion, (c) network edges that failed
to transmit the contagion, and (d) $\varepsilon$-edges that failed to transmit
the contagion.

\begin{figure}[t]
  \subfigure[Graph $G$ on five vertices and four network edges (solid edges). $\varepsilon$-edges shown as dashed lines.]{\includegraphics[width=0.4\textwidth]{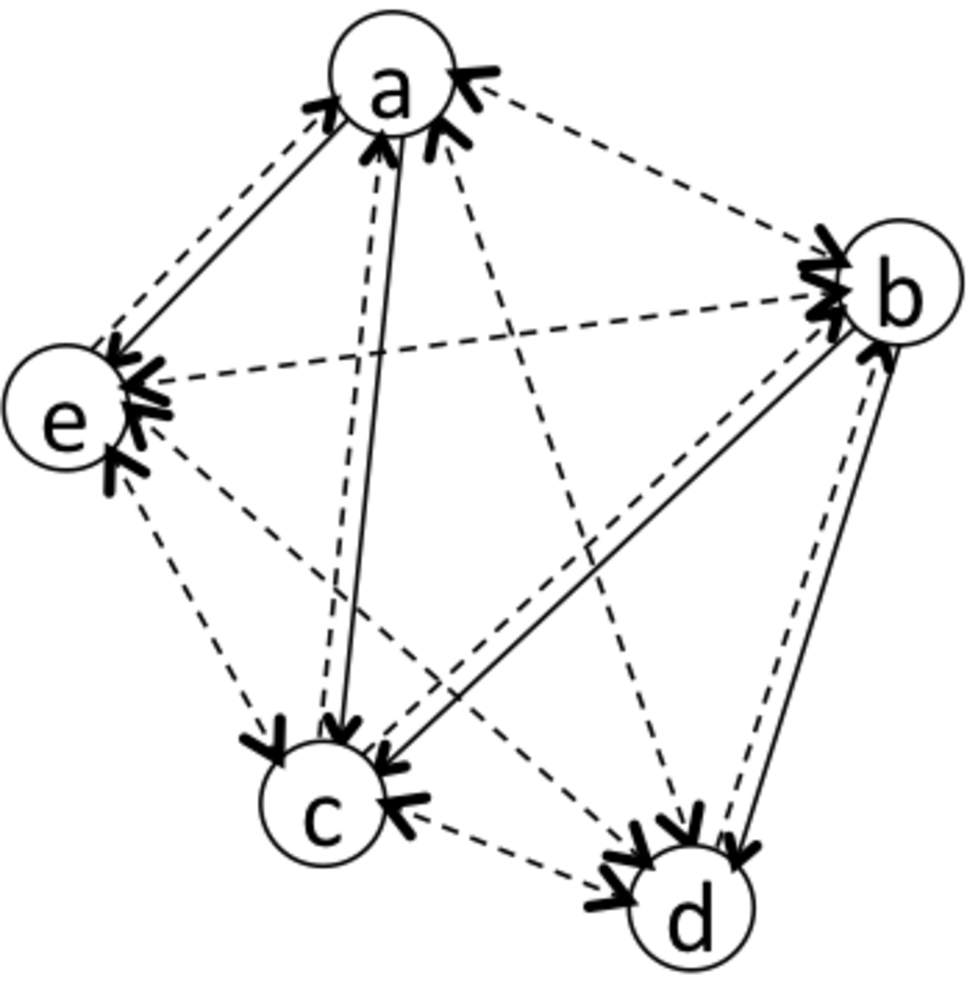}
    \label{fig:epsedges1}} \hspace{2cm}
  \subfigure[Cascade propagation tree $T=\{(a,b), (b,c), (b,d)\}$]{\includegraphics[width=0.4\textwidth]{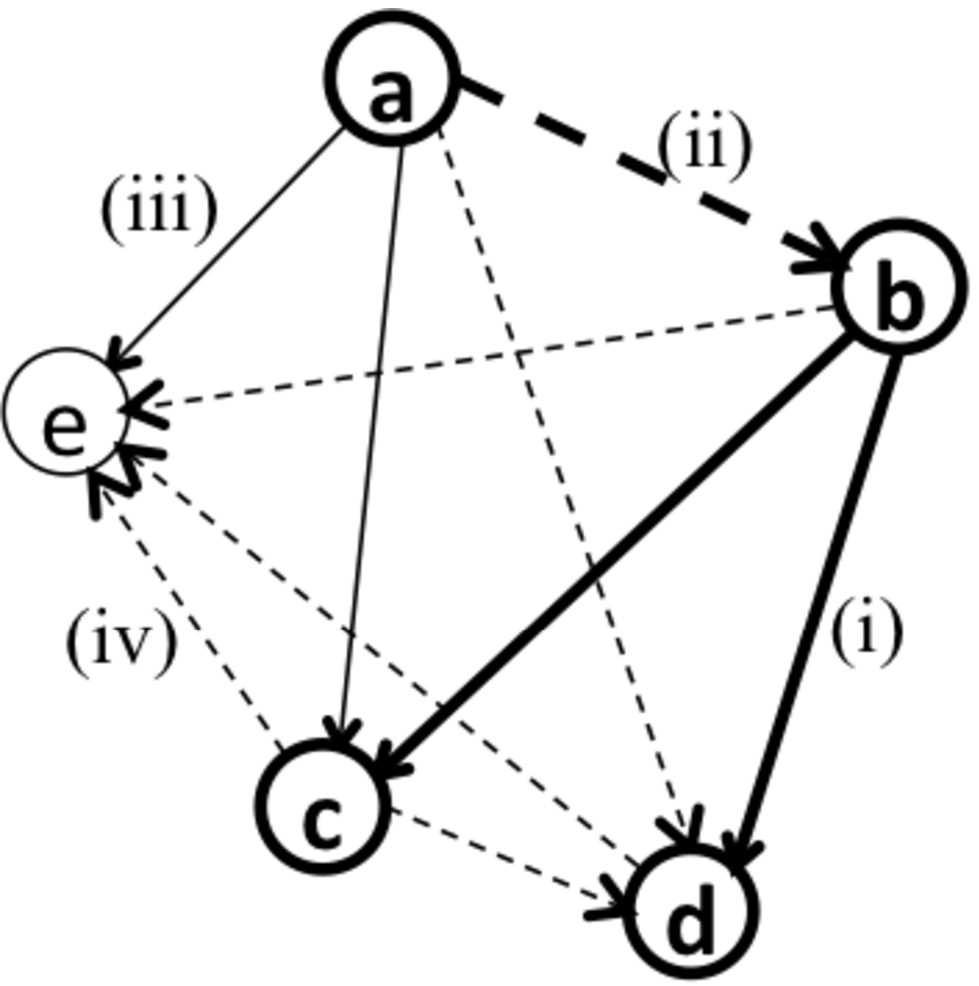}
    \label{fig:epsedges2}}
  \caption{(a) Graph $G$: Network edges $E$ are shown as solid, and $\varepsilon$-edges are shown as dashed lines.
  (b) Propagation tree $T=\{(a,b), (b,c), (b,d)\}$. Four types of edges are labeled: (i) network edges that
  transmitted the contagion (solid bold), (ii) $\varepsilon$-edges that transmitted the contagion (dashed bold),
  (iii) network edges that failed to transmit the contagion (solid), and (iv) $\varepsilon$-edges that failed to transmit the 
  contagion (dashed).}
  \label{fig:epsedges}
\end{figure}

We can now rewrite the cascade likelihood $P(c | T)$ as combination of products
of edge-types and the product over the edge incubation times:

\begin{align}
  P(c|T) &=& \beta^q \ \varepsilon^{q'} \ (1-\beta)^s \ (1-\varepsilon)^{s'} \prod_{(u,v)\in E_T} P_c(v, u) \\
  &\approx& \beta^q \ \varepsilon^{q'} \ (1-\varepsilon)^{s+s'} \prod_{(u,v)\in E_T} P_c(v, u),
  \label{eq:probcasc3}
\end{align}
where $q$ is the number of network edges in $T$ (type (a) edges in
Fig.~\ref{fig:epsedges2}), $q'$ is the number of $\varepsilon$-edges in $T$,
$s$ is the number of network edges that did not transmit and $s'$ is the number
of $\varepsilon$-edges that did not transmit. Note that the above approximation
is valid since real networks are sparse and cascades are generally small, and hence
$s' \gg s$. Thus, even though $\beta \gg \varepsilon$ we expect $(1-\beta)^s$
to be of about same order of magnitude as $(1-\varepsilon)^{s'}$.

The formulation in Equation~\ref{eq:probcasc3} has several benefits. Due to the
introduction of $\varepsilon$-edges the likelihood $P(c|T)$ is always positive. For
example, even if we consider graph $G$ with no edges, $P(c|T)$ is still well
defined as we can explain the tree $T$ via the diffusion over the
$\varepsilon$-edges. A second benefit that will become very useful later is that
the likelihood now becomes monotonic in the network edges of $G$. This means
that adding an edge to $G$ (\ie, converting $\varepsilon$-edge into a network
edge) only increases the likelihood.

\xhdr{Considering only the most likely propagation tree} So far we introduced
the concept of $\varepsilon$-edges to model the external influence or diffusion
that is exogenous to the network, and introduce an approximation to treat all
edges that did not participate in the diffusion as $\varepsilon$-edges.

Now we consider the last approximation, where instead of considering all
possible cascade propagation trees $T$, we only consider the most likely
cascade propagation trees $T$:

\begin{align}
  P(C | G) = \prod_{c \in C} \ \sum_{T \in \mathcal{T}_c(G)} P(c|T)
  \approx \prod_{c \in C} \ \max_{T \in \mathcal{T}_c(G)} P(c|T).
  \label{eq:pcasc2}
\end{align}

Thus now we aim to solve the network inference problem by finding a graph $G$
that maximizes Equation~\ref{eq:pcasc2}, where $P(c|T)$ is defined in
Equation~\ref{eq:probcasc3}.

This formulation simplifies the original network inference
problem by considering the most likely (\emph{best}) propagation tree $T$ per
cascade $c$ instead of considering all possible propagation trees $T$ for each
cascade $c$. Although in some cases we expect the likelihood of $c$ with respect
to the true tree $T'$ to be much higher than that with respect to any competing tree 
$T''$ and thus the probability mass will be concentrated at $T'$, there might be some 
cases in which the probability mass does not concentrate on one particular T. However,
we run extensive experiments on small networks with different structures in which both 
the original network inference problem and the alternative formulation can be solved 
using exhaustive search. Our experimental results looked really similar and the results were 
indistinguishable. Therefore, we consider our approximation to work well in practice.

For convenience, we work with the log-likelihood $\log P(c|T)$ rather than
likelihood $P(c|T)$. Moreover, instead of directly maximizing the
log-likelihood we equivalently maximize the following objective function that defines the
improvement of log-likelihood for cascade $c$ occurring in graph $G$ over $c$
occurring in an empty graph $\bar{K}$ (\ie, graph with only $\varepsilon$-edges
and no network edges):

\begin{equation}
  F_c(G) = \max_{T \in \mathcal{T}_c(G)} \log P(c | T) - \max_{T \in \mathcal{T}_c(\bar{K})} \log P(c | T).
\label{eq:pcasc3}
\end{equation}

Maximizing Equation~\eqref{eq:pcasc2} is equivalent to maximizing the following log-likelihood function:
\begin{equation}
  F_C(G) = \sum_{c \in C} F_c(G).
\label{eq:pcasc4}
\end{equation}

We now expand Eq.~\eqref{eq:pcasc3} and obtain an instance of a {\em simplified
diffusion network inference problem}:

\begin{equation}
  \hat{G} =  \arg \max_G F_C(G) = \sum_{c \in C} \max_{T \in \mathcal{T}_c(G)} \sum_{(i,j)\in E_T} w_{c}(i,j),
  \label{eq:pcasc5}
\end{equation}
where $w_{c}(i,j)= \log P_c'(i, j) -\log \varepsilon$ is a non-negative weight
which can be interpreted as the improvement in log-likelihood of edge $(i,j)$
under the most likely propagation tree $T$ in $G$. Note that by the
approximation in Equation~\ref{eq:probcasc3} one can ignore the contribution of
edges that did not participate in a particular cascade $c$. The contribution of
these edges is constant, \ie, independent of the particular shape that propagation
tree $T$ takes. This is due to the fact that each spanning tree $T$ of $G$ with
node set $V_T$ has $|V_T|-1$ (network and $\varepsilon$-) edges that
participated in the cascade, and all remaining edges stopped the cascade from
spreading. The number of non-spreading edges depends only on the node set $V_T$
but {\em not} the edge set $E_T$. Thus, the tree $T$ that maximizes $P(c | T)$ also
maximizes $\sum_{(i,j)\in E_T} w_{c}(i,j)$.

Since $T$ is a tree that maximizes the sum of the edge weights this means that
the most likely propagation tree $T$ is simply the {\em maximum weight directed
spanning tree} of nodes $V_T$, where each edge $(i,j)$ has weight $w_{c}(i,j)$,
and $F_c(G)$ is simply the sum of the weights of edges in $T$.

We also observe that since edges $(i,j)$ where $t_{i} \geq t_{j}$ have weight 0
(\emph{i.e.}, such edges are not present) then the outgoing edges of any node
$u$ only point forward in time, \ie, a node can not infect already infected
nodes. Thus for a fixed cascade $c$, the collection of edges with positive
weight forms a directed \emph{acyclic} graph (DAG).

Now we use the fact that the collection of edges with positive weights forms a
directed acyclic graph by observing that the maximum weight directed spanning
tree of a DAG can be computed efficiently:

\begin{proposition}
\label{prop:dagtree}
  In a DAG $D(V,E,w)$ with vertex set $V$ and nonnegative edge weights $w$, the maximum
  weight directed spanning tree can be found by choosing, for each node $v$, an
  incoming edge $(u,v)$ with maximum weight $w(u,v)$.
\end{proposition}

\begin{proof}
The score
$$
  S(T) = \sum_{(i,j)\in T} w(i,j)=\sum_{i\in V} w(Par_{T}(i),i)
$$
of a tree $T$ is the sum of the incoming edge weights $w(Par_{T}(i),i)$ for
each node $i$, where $Par_{T}(i)$ is the parent of node $i$ in $T$ (and the
root is handled appropriately).  Now,
$$
  \max_{T} S(T) =\max_{T} \sum_{(i,j)\in T} w(i,j)=\sum_{i\in V} \max_{Par_{T}(i)}w(Par_{T}(i),i).
$$

Latter equality follows from the fact that since $G$ is a DAG, the maximization
can be done independently for each node without creating any cycles.
\end{proof}

This proposition is a special case of the more general maximum spanning tree
(MST) problem in directed graphs~\cite{edmonds67branchings}. The important fact
now is that we can find the best propagation tree $T$ in time $O(|V_T| D_{in})$, i.e., linear
in the number of edges and the maximum in-degree $D_{in}=\max_{u\in V_{T}}d_{in}(u)$ by simply 
selecting an incoming edge of highest weight for each node $u \in V_T$. Algorithm~\ref{alg:dagtree} 
provides the pseudocode to efficiently compute the maximum weight directed spanning tree of a DAG.

\begin{algorithm}[t]
\caption{Maximum~weight directed spanning tree of a DAG}
  \label{alg:dagtree}
  \begin{algorithmic}
  \REQUIRE Weighted directed acyclic graph $D(V,E,w)$
  \STATE $T \leftarrow \{\}$

  \FORALL{$i \in V$}
  \STATE $Par_T(i) = \arg \max_j w(j,i)$
  \STATE $T \leftarrow T\cup\{(Par_T(i),j)\}$
  \ENDFOR

  \RETURN $T$
\end{algorithmic}
\end{algorithm}

Putting it all together we have shown how to efficiently evaluate the log-likelihood $F_C(G)$ of 
a graph $G$. To find the most likely tree $T$ for a single cascade takes $O(|V_T| D_{in})$, and 
this has to be done for a total of $|C|$ cascades. Interestingly, this is independent of the size 
of graph $G$ and only depends on the amount of observed data (\ie, size and the number of cascades).

\subsection{The \netinf algorithm for efficient maximization of $\bf F_C(G)$}

Now we aim to find graph $G$ that maximizes the log-likelihood $F_C(G)$. First
we notice that by construction $F_{C}(\bar{K})=0$, \emph{i.e.}, the empty graph
has score 0. Moreover, we observe that the objective function $F_{C}$ is
non-negative and monotonic. This means that $F_{C}(G)\leq F_{C}(G')$ for graphs $G(V,E)$ and $G'(V, E')$,
where $E \subseteq E'$. Hence adding more edges to $G$ does not decrease the
solution quality, and thus the complete graph maximizes $F_{C}$.  Monotonicity can be 
shown by observing that, as edges are added to $G$, $\varepsilon$-edges are converted 
to network edges, and therefore the weight of any tree (and therefore the value of the 
maximum spanning tree) can only increase.
However, since real-world social and information networks are usually sparse, we are
interested in inferring a \emph{sparse} graph $G$, that only contains some
small number $k$ of edges. Thus we aim to solve:

\begin{problem}
  Given the infection times of a set of cascades $C$, probability of propagation $\beta$
  and the incubation time distribution $P_{c}(i,j)$, find $\hat G$ that
  maximizes:
  \begin{equation}
      G^{*}=\argmax_{|G|\leq \nedge}F_{C}(G),
  \label{eq:maxfc2}
  \end{equation}
  where the maximization is over all graphs $G$ of at most $\nedge$ edges, and $F_{C}(G)$
  is defined by Eqs.~\ref{eq:pcasc4} and~\ref{eq:pcasc5}.
\end{problem}

Naively searching over all $k$ edge graphs would take time exponential in $k$,
which is intractable. Moreover, finding the optimal solution to
Eq.~\eqref{eq:maxfc2} is NP-hard, so we cannot expect to find the optimal
solution:

\begin{theorem}\label{thm:hardness}
  The network inference problem defined by equation~\eqref{eq:maxfc2} is NP-hard.
\end{theorem}

\begin{proof}
By reduction from the MAX-$k$-COVER problem~\cite{khuller1999budgeted}.  In
MAX-$k$-COVER, we are given a finite set $W$, $|W|=n$ and a collection of
subsets $S_{1},\dots,S_{m}\subseteq W$.  The function $$F_{MC}(A)=|\cup_{i\in
A}S_{i}|$$ counts the number of elements of $W$ covered by sets indexed by $A$.
Our goal is to pick a collection of $k$ subsets $A$ maximizing $F_{MC}$.  We
will produce a collection of $n$ cascades $C$ over a graph $G$ such that
$\max_{|G|\leq k}F_{C}(G) = \max_{|A|\leq k}F_{MC}(A)$.  Graph $G$ will be
defined over the set of vertices $V = \{1,\dots,m\}\cup\{r\}$, \emph{i.e.}, there is
one vertex for each set $S_{i}$ and one extra vertex $r$.  For each element
$s\in W$ we define a cascade which has time stamp $0$ associated with all nodes $i\in V$
such that $s\in S_{i}$, time stamp $1$ for node $r$ and $\infty$ for all
remaining nodes.

Furthermore, we can choose the transmission model such that $w_{c}(i,r)=1$
whenever $s\in S_{i}$ and $w_{c}(i',j')=0$ for all remaining edges $(i',j')$,
by choosing the parameters $\varepsilon$, $\alpha$ and $\beta$ appropriately.
Since a directed spanning tree over a graph $G$ can contain at most one edge
incoming to node $r$, its weight will be $1$ if $G$ contains any edge from a
node $i$ to $r$ where $s\in S_{i}$, and $0$ otherwise.  Thus, a graph $G$
of at most $k$ edges corresponds to a feasible solution $A_{G}$ to
MAX-$k$-COVER where we pick sets $S_{i}$ whenever edge $(i,r)\in G$, and each
solution $A$ to MAX-$k$-COVER corresponds to a feasible solution $G_{A}$ of
\eqref{eq:maxfc2}. Furthermore, by construction, $F_{MC}(A_{G})=F_{C}(G)$.
Thus, if we had an efficient algorithm for deciding whether there exists a
graph $G$, $|G| \leq k$ such that $F_{C}(G)> c$, we could use the algorithm to
decide whether there exists a solution $A$ to MAX-$k$-COVER with value at least
$c$.
\end{proof}

While finding the optimal solution is hard, we now show that $F_{C}$ satisfies
\emph{submodularity}, a natural diminishing returns property. The submodularity
property allows us to efficiently find a \emph{provably near-optimal} solution
to this otherwise NP-hard optimization problem.

A set function $F: 2^{W}\rightarrow\mathbb{R}$ that maps subsets of a finite
set $W$ to the real numbers is \emph{submodular} if for $A\subseteq B\subseteq
W$ and $s\in W\setminus B$, it holds that

$$
F(A\cup\{s\})-F(A)\geq F(B\cup\{s\})-F(B).
$$

This simply says adding $s$ to the set $A$ increases the score more than adding
$s$ to set $B$ ($A \subseteq B$).

Now we are ready to show the following result that enables us to find a near
optimal network $G$:

\begin{theorem}\label{thm:submodular}
  Let $V$ be a set of nodes, and $C$ be a collection of cascades hitting the
  nodes $V$. Then $F_{C}(G)$ is a submodular function $F_{C}:2^{W}\rightarrow
  \mathbb{R}$ defined over subsets $W\subseteq V\times V$ of directed edges.
\end{theorem}

\begin{proof}
Fix a cascade $c$, graphs $G\subseteq G'$ and an edge $e=(r,s)$ not contained
in $G'$.  We will show that $F_{c}(G\cup\{e\})-F_{c}(G)\geq
F_{c}(G'\cup\{e\})-F_{c}(G')$.  Since nonnegative linear combinations of
submodular functions are submodular, the function $F_{C}(G)=\sum_{c\in
C}F_{c}(G)$ is submodular as well.  Let $w_{i,j}$ be the weight of edge $(i,j)$
in $G\cup\{e\}$, and $w'_{i,j}$ be the weight in $G'\cup\{e\}$.  As argued
before, the maximum weight directed spanning tree for DAGs is obtained by
assigning to each node the incoming edge with maximum weight.  Let $(i,s)$ be
the edge incoming at $s$ of maximum weight in $G$, and $(i',s)$ the maximum
weight incoming edge in $G'$.  Since $G\subseteq G'$ it holds that $w_{i,s}\leq
w'_{i',s}$.  Furthermore, $w_{r,s}=w'_{r,s}$. Hence,
\begin{align*}
  F_{c}(G\cup\{(r,s)\})-F_{c}(G)&=\max(w_{i,s},w_{r,s})-w_{i,s}\\
  &\geq \max(w'_{i',s},w'_{r,s})-w'_{i',s}\\
  &=F_{c}(G'\cup\{(r,s)\})-F_{c}(G'),
\end{align*}
proving submodularity of $F_{c}$.
\end{proof}

Maximizing submodular functions in general is NP-hard~\cite{khuller1999budgeted}. A
commonly used heuristic is the \emph{greedy algorithm}, which starts with an empty
graph $\bar{K}$, and iteratively, in step $i$, adds the edge $e_i$ which maximizes
the marginal gain:

\begin{equation}
  e_i = \argmax_{e \in G \backslash G_{i-1}} F_C(G_{i-1} \cup \{e\}) - F_C(G_{i-1}).
  \label{eq:pmarginal}
\end{equation}

The algorithm stops once it has selected $\nedge$ edges, and returns the
solution $\hat{G}=\{e_{1},\dots,e_{\nedge}\}$. The stopping criteria,
\emph{i.e.}, value of $\nedge$, can be based on some threshold of the marginal
gain, of the number of estimated edges or another more sophisticated heuristic.

In our context we can think about the greedy algorithm as starting on an empty
graph $\bar{K}$ with no network edges. In each iteration $i$, the algorithm adds to $G$
the edge $e_i$ that currently improves the most the value of the log-likelihood. Another way to 
view the greedy algorithm is that it starts on a fully connected graph $\bar{K}$ where all 
the edges are $\varepsilon$-edges. Then adding an edge to graph $G$ corresponds to 
that edge changing the type from $\varepsilon$-edge to a network edge. Thus our algorithm 
iteratively swaps $\varepsilon$-edges to network edges until $k$ network edges have been 
swapped (\ie, inserted into the network $G$).

\xhdr{Guarantees on the solution quality}
Considering the NP-hardness of the problem, we might expect the greedy
algorithm to perform arbitrarily bad. However, we will see that this is not the
case. A fundamental result of Nemhauser et al.~\cite{nemhauser1978analysis}
proves that for monotonic submodular functions, the set $\hat{G}$ returned by
the greedy algorithm obtains at least a constant fraction of $(1-1/e)\approx
63\%$ of the optimal value achievable using $\nedge$ edges.

Moreover, we can acquire a tight \emph{online} data-dependent bound on the
solution quality:

\begin{theorem}[\cite{leskovec2007cost}]
  For a graph $\hat G$, and each edge $e \notin \hat G$, let $\delta_{e} = F_C(\hat G \cup \{e\}) - F_C(\hat G)$. Let $e_1, \ldots e_B$
  be the sequence with $\delta_{e}$ in decreasing order, where $B$ is the total number of edges with marginal gain greater than $0$. 
  Then,
  $$
  \max_{|G|\leq \nedge} F_c(G) \leq F_c(\hat G) + \sum_{i=1}^{\nedge} \delta_{e_i}.
  $$
\label{th:online_bound}
\end{theorem}

Theorem~\ref{th:online_bound} computes how far a given $\hat G$ (obtained by \emph{any}
algorithm) is from the unknown NP-hard to find optimum.

\xhdr{Speeding-up the \netinf algorithm} To make the algorithm scale to
networks with thousands of nodes we speed-up the algorithm by several orders of
magnitude by considering two following two improvements:

\noindent{{\em Localized update:}} Let $C_i$ be the subset of cascades that go through the
node $i$ (\ie, cascades in which node $i$ is infected). Then consider that in some step $n$ 
the \emph{greedy algorithm} selects the network edge $(j,i)$ with marginal gain $\delta_{j, i}$,
and now we have to update the optimal tree of each cascade. We make a simple observation
that adding the network edge $(j, i)$ may only change the optimal trees of the cascades in the set $C_i$ 
and thus we only need to revisit (and potentially update) the trees of cascades in $C_i$. Since cascades
are local (\ie, each cascade hits only a relatively small subset of the network), this localized 
updating procedure  speeds up the algorithm considerably.

\noindent{{\em Lazy evaluation:}} It can be used to drastically reduce the number of
evaluations of marginal gains $F_{C}(G\cup\{e\})-F_{C}(G)$~\cite{leskovec2007cost}.
This procedure relies on the submodularity of $F_{C}$. The key idea behind lazy evaluations 
is the following. Suppose $G_{1},\dots,G_{k}$ is the sequence of graphs produced during iterations 
of the greedy algorithm. Now let us consider the marginal gain $$\Delta_{e}(G_{i})=F_{C}(G_{i}\cup\{e\})-F_{C}(G_{i})$$ 
of adding some edge $e$ to any of these graphs. Due to the submodularity of $F_{C}$ it holds that 
$\Delta_{e}(G_{i})\geq \Delta_{e}(G_{j})$ whenever $i\leq j$. Thus, the marginal gains of $e$ can 
only monotonically decrease over the course of the greedy algorithm. This means that elements which 
achieve very little marginal gain at iteration $i$ cannot suddenly produce large marginal gain at 
subsequent iterations.  This insight can be exploited by maintaining a priority queue data structure 
over the edges and their respective marginal gains.  At each iteration, the greedy algorithm retrieves 
the highest weight (priority) edge.  Since its value may have decreased from previous iterations, it 
recomputes its marginal benefit.  If the marginal gain remains the same after recomputation, it has to be 
the edge with highest marginal gain, and the greedy algorithm will pick it.  If it decreases, one reinserts 
the edge with its new weight into the priority queue and continues. Formal details and pseudo-code can be 
found in \cite{leskovec2007cost}.

As we will show later, these two improvements decrease the run time by several
orders of magnitude with {\em no loss} in the solution quality. We call the
algorithm that implements the greedy algorithm on this alternative formulation
with the above speedups the \netinf algorithm (Algorithm~\ref{alg:netinf}). In
addition, \netinf nicely lends itself to parallelization as likelihoods of
individual cascades and likelihood improvements of individual new edges can
simply be computed independently. This allows us to to tackle even bigger
networks in shorter amounts of time.

A space and runtime complexity analysis of \netinf depends heavily of the structure of the network, and therefore 
it is necessary to make strong assumptions on the structure. Due to this, it is out of the scope of the paper to include 
a formal complexity analysis. Instead, we include an empirical runtime analysis in the following section.

\begin{algorithm}[t]
\caption{The \netinf Algorithm\label{alg:netinf}}
\begin{algorithmic}
  \REQUIRE Cascades and hit times $C=\{(c, \bt_c)\}$, number of edges $\nedge$

  \STATE $G \leftarrow \bar{K}$
  \FORALL{$c \in C$}
    \STATE $T_c \leftarrow dag\_tree(c)$ \hfill \COMMENT {Find most likely tree (Algorithm~\ref{alg:dagtree})}
  \ENDFOR

  \WHILE{$|G| < \nedge$}
    \FORALL{$(j,i) \notin G$}
      \STATE $\delta_{j,i} = 0$ \hfill \COMMENT{Marginal improvement of adding edge $(j,i)$ to G}
      \STATE $M_{j,i} \leftarrow \emptyset$
      \FORALL{$c : t_j < t_i$ in $c$} 
        \STATE Let $w_c(m,n)$ be the weight of $(m,n)$ in $G \cup \{(j,i)\}$
        \IF{$w_c(j,i) \geq w_c(Par_{T_c}(i), i)$}
          \STATE $\delta_{j,i} = \delta_{j,i} + w_c(j,i)- w_c(Par_{T_c}(i), i)$
          \STATE $M_{j,i} \leftarrow M_{j,i} \cup \{c\}$
        \ENDIF
      \ENDFOR
    \ENDFOR

    \STATE $(j^{*},i^{*}) \leftarrow \arg \max_{(j,i) \in C \backslash G} \delta_{j,i}$
    \STATE $G \leftarrow G \cup \{(j^{*},i^{*})\}$
    \FORALL{$c \in M_{j^{*},i^{*}}$}
      \STATE $Par_{T_c}(i^{*}) \leftarrow j^{*}$
    \ENDFOR
  \ENDWHILE

\RETURN G;
\end{algorithmic}
\end{algorithm}

\section{Experimental evaluation}
\label{sec:evaluation}
In this section we proceed with the experimental evaluation of our proposed \netinf
algorithm for inferring network of diffusion. We analyze the performance of
\netinf on synthetic and real networks. We show that our algorithm performs
surprisingly well, outperforms a heuristic baseline and correctly discovers
more than 90\% of the edges of a typical diffusion network.

\subsection{Experiments on synthetic data}

The goal of the experiments on synthetic data is to understand how the
underlying network structure and the propagation model (exponential and
power-law) affect the performance of our algorithm. The second goal is to
evaluate the effect of simplification we had to make in order to arrive to an
efficient network inference algorithm. Namely, we assume the contagion
propagates in a tree pattern $T$ (\ie, exactly $E_T$ edges caused the
propagation), consider only the most likely tree $T$ (Eq.~\ref{eq:pcasc2}), and treat 
non-propagating network edges as $\varepsilon$-edges (Eq.~\ref{eq:probcasc3}).
\begin{figure}
\centering
  \subfigure[FF: Cascades per edge]{\includegraphics[width=0.48\textwidth]{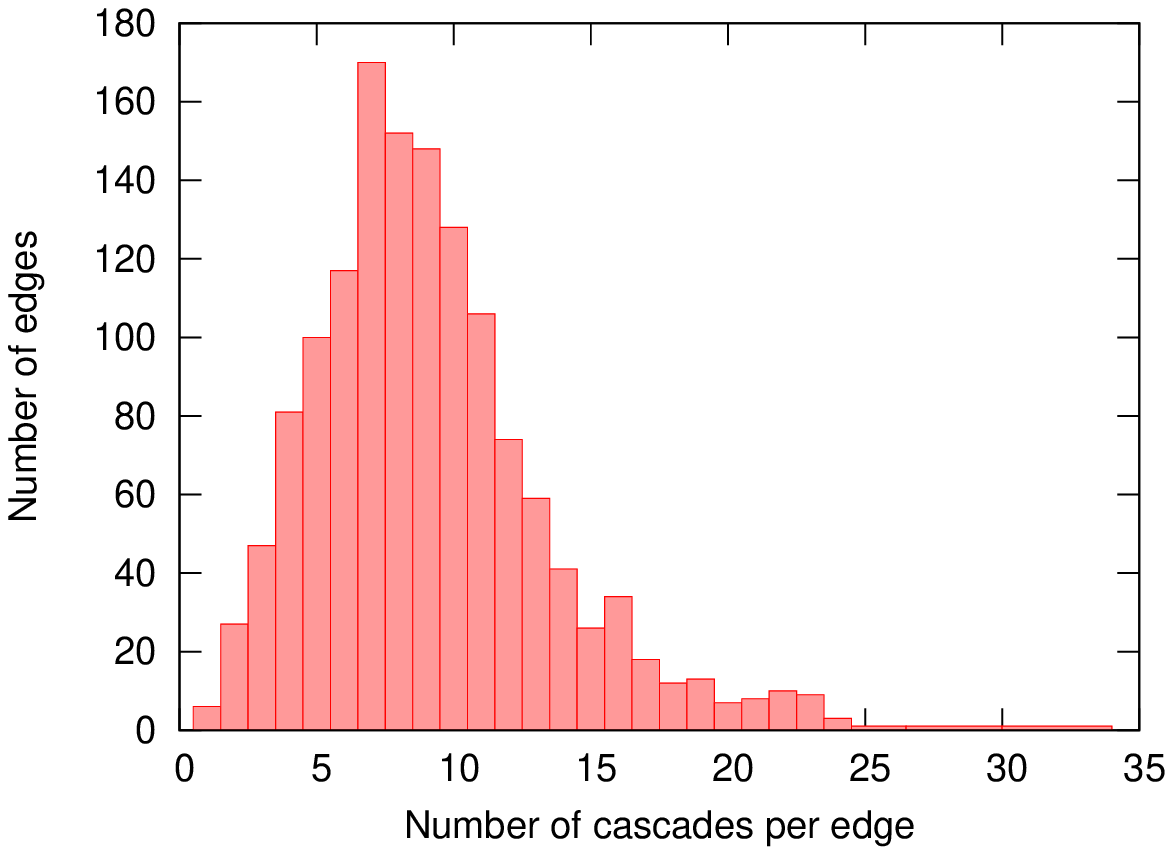} \label{fig:NumCascadesPerEdge}}
  \subfigure[FF: Cascade size]{\includegraphics[width=0.48\textwidth]{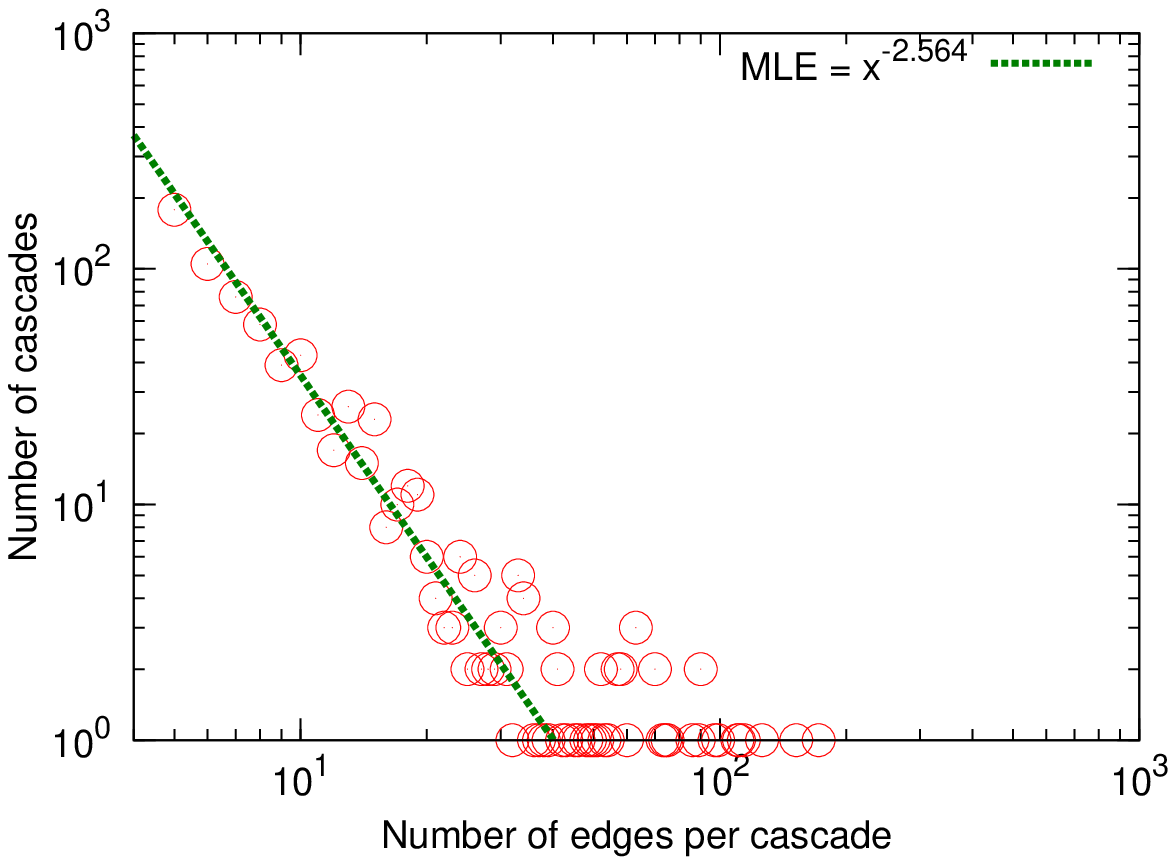} \label{fig:SizeCascades}}
  \caption{Number of cascades per edge and cascade sizes for a Forest Fire network ($1,024$ nodes, $1,477$ edges) with forward burning
  probability $0.20$, backward burning probability $0.17$ and exponential incubation time model with parameter $\alpha=1$ and propagation 
  probability $\beta=0.5$. The cascade size distribution follows a power-law. We found the power-law coefficient using maximum likelihood estimation 
  (MLE).
}
\label{fig:StatsCascades}
\end{figure}

In general, in all our experiments we proceed as follows: We are given a true
diffusion network $G^*$, and then we simulate the propagation of a set of
contagions $c$ over the network $G^*$. Diffusion of each contagion creates a
cascade and for each cascade, we record the node hit times $t_u$. Then, given
these node hit times, we aim to recover the network $G^*$ using the \netinf
algorithm. For example, Figure~\ref{fig:FFN20}(a) shows a graph $G^*$ of 20
nodes and 23 directed edges. Using the exponential incubation time model and
$\beta=0.2$ we generated $24$ cascades. Now given the node infection times, we
aim to recover $G^*$. A baseline method (b) (described below) performed poorly
while \netinf (c) recovered $G^*$ almost perfectly by making only two errors
(red edges).

\xhdr{Experimental setup} Our experimental methodology is composed of the
following steps:

\begin{enumerate}
  \item Ground truth graph $G^*$
  \item Cascade generation: Probability of propagation $\beta$, and the
      incubation time model with parameter $\alpha$.
  \item Number of cascades
\end{enumerate}

\noindent{\em (1) Ground truth graph $G^*$:} We consider two models of directed
real-world networks to generate $G^*$, namely, the Forest Fire model~\cite{leskovec2005graphs}
and the Kronecker Graphs model~\cite{leskovec2007scalable}.
For Kronecker graphs, we consider three sets of parameters that produce
networks with a very different global network structure: a random
graph~\cite{erdos60random} (Kronecker parameter matrix $[0.5, 0.5; 0.5, 0.5]$), a core-periphery 
network~\cite{jure08ncp} ($[0.962, 0.535; 0.535, 0.107]$) and a network with hierarchical community 
structure~\cite{clauset08hierarchical} ($[0.962, 0.107; 0.107, 0.962]$). The Forest Fire generates networks 
with power-law degree distributions that follow the densification power
law~\cite{barabasi99emergence,jure07evolution}.

\noindent{\em (2) Cascade propagation:} We then simulate cascades on $G^*$
using the generative model defined in Section~\ref{sec:problem}. For the
simulation we need to choose the incubation time model (\ie, power-law or
exponential and parameter $\alpha$). We also need to fix the parameter $\beta$, that
controls probability of a cascade propagating over an edge. Intuitively,
$\alpha$ controls how fast the cascade spreads (\ie, how long the incubation
times are), while $\beta$ controls the size of the cascades. Large $\beta$
means cascades will likely be large, while small $\beta$ makes most of
the edges fail to transmit the contagion which results in small infections.

\noindent{\em (3) Number of cascades:} Intuitively, the more data our algorithm
gets the more accurately it should infer $G^*$. To quantify the amount of data
(number of different cascades) we define $E_l$ to be the set of edges that
participate in at least $l$ cascades. This means $E_l$ is a set of edges that
transmitted at least $l$ contagions. It is important to note that if an edge of
$G^*$ did not participate in any cascade (\ie, it never transmitted a
contagion) then there is no trace of it in our data and thus we have no chance
to infer it. In our experiments we choose the minimal amount of data (\ie,
$l=1$) so that we at least in principle could infer the true network $G^*$.
Thus, we generate as many cascades as needed to have a set $E_l$ that contains
a fraction $f$ of all the edges of the true network $G^*$. In all our
experiments we pick cascade starting nodes uniformly at random and generate
enough cascades so that 99\% of the edges in $G^*$ participate in at least one
cascade, \emph{i.e.}, 99\% of the edges are included in $E_1$.

Table~\ref{tab:NumCascades} shows experimental values of number of cascades
that let $E_1$ cover different percentages of the edges. To have a closer look
at the cascade size distribution, for a Forest Fire network on 1,024 nodes and
1,477 edges, we generated 4,038 cascades. The majority of edges took part in 4
to 12 cascades and the cascade size distribution follows a power law
(Figure~\ref{fig:SizeCascades}). The average and median number of cascades per
edge are 9.1 and 8, respectively (Figure~\ref{fig:NumCascadesPerEdge}).

\begin{table}[t]
    \small
    \centering
    \begin{tabular}{l c c c c r}
       \textbf{Type of network} & \textbf{$\bf f$} & \textbf{$\bf |C|$} & \textbf{$\bf r$} & \textbf{BEP} & \textbf{AUC} \\
       \hline
        \multirow{5}{*}{Forest Fire} & 0.5 & 388 & 2,898 & 0.393 & 0.29\\
        & 0.9 & 2,017 & 14,027 & 0.75 & 0.67\\
        & 0.95 & 2,717 & 19,418 & 0.82 & 0.74\\
        & 0.99 & 4,038 & 28,663 & 0.92 & 0.86\\
        \hline
        \multirow{5}{*}{Hierarchical Kronecker} & 0.5 & 289 & 1,341 & 0.37 & 0.30\\
        & 0.9 & 1,209 & 5,502 & 0.81 & 0.80\\
        & 0.95 & 1,972 & 9,391 & 0.90 & 0.90\\
        & 0.99 & 5,078 & 25,643 & 0.98 & 0.98\\
        \hline
        \multirow{5}{*}{Core-periphery Kronecker} & 0.5 & 140 & 1,392 & 0.31 & 0.23\\
        & 0.9 & 884 & 9,498 & 0.84 & 0.80\\
        & 0.95 & 1,506 & 14,125 & 0.93 & 0.91\\
        & 0.99 & 3,110 & 30,453 & 0.98 & 0.96\\
        \hline
        \multirow{5}{*}{Flat Kronecker} & 0.5 & 200 & 1,324 & 0.34 & 0.26\\
        & 0.9 & 1,303 & 7,707 & 0.84 & 0.83\\
        & 0.95 & 1,704 & 9,749 & 0.89 & 0.88\\
        & 0.99 & 3,652 & 21,153 & 0.97 & 0.97
    \end{tabular}
    \caption{Performance of synthetic data. Break-even Point (BEP) and Receiver Operating Characteristic (AUC)
    when we generated the minimum number of $|C|$ cascades so that $f$-fraction
    of edges participated in at least one cascades $|E_l|\ge f|E|$. These $|C|$
    cascades generated the total of $r$ edge transmissions, \ie, average
    cascade size is $r/|C|$. All networks have 1,024 nodes and 1,446 edges. We use the exponential incubation 
    time model with parameter $\alpha=1$, and in each case we set the probability $\beta$ such that $r/|C|$ is
    neither too small nor too large (\ie, $\beta \in (0.1, 0.6)$).}
    \label{tab:NumCascades}
\end{table}

\xhdr{Baseline method} To infer a diffusion network $\hat{G}$, we consider the
a simple baseline heuristic where we compute the score of each edge and then
pick $\nedge$ edges with highest score.

More precisely, for each \emph{possible} edge $(u,v)$ of $G$, we compute
$w(u,v) = \sum_{c \in C} P_c(u,v)$, \emph{i.e.}, overall how likely were the
cascades $c \in C$ to propagate over the edge $(u,v)$. Then we simply pick the
$\nedge$ edges $(u,v)$ with the highest score $w(u,v)$ to obtain $\hat{G}$. For
example, Figure~\ref{fig:FFN20}(b) shows the results of the baseline method on
a small graph.

\xhdr{Solution quality} We evaluate the performance of the \netinf algorithm in
two different ways. First, we are interested in how successful \netinf is at
optimizing the objective function $F_C(G)$ that is NP-hard to optimize exactly.
Using the online bound in Theorem~\ref{th:online_bound}, we can assess at most
how far from the unknown optimal the \netinf solution is in terms of the
log-likelihood score. Second, we also evaluate the \netinf based on accuracy,
\ie, what fraction of edges of $G^*$ \netinf managed to infer correctly.

Figure~\ref{fig:ScoreKROH1000NGraph} plots the value of the log-likelihood
improvement $F_C(G)$ as a function of the number of edges in $G$. In red we
plot the value achieved by \netinf and in green the upper bound using
Theorem~\ref{th:online_bound}. The plot shows that the value of the unknown
optimal solution (that is NP-hard to compute exactly) is somewhere between the
red and the green curve. Notice that the band between two curves, the optimal
and the \netinf curve, is narrow. For example, at 2,000 edges in $\hat{G}$,
\netinf finds the solution that is least 97\% of the optimal graph. Moreover,
also notice a strong diminishing return effect. The value of the objective
function flattens out after about 1,000 edges. This means that, in practice,
very sparse solutions (almost tree-like diffusion graphs) already achieve very
high values of the objective function close to the optimal.

\begin{figure}[t]
  \centering
  \subfigure[Kronecker network]{\includegraphics[width=0.8\textwidth]{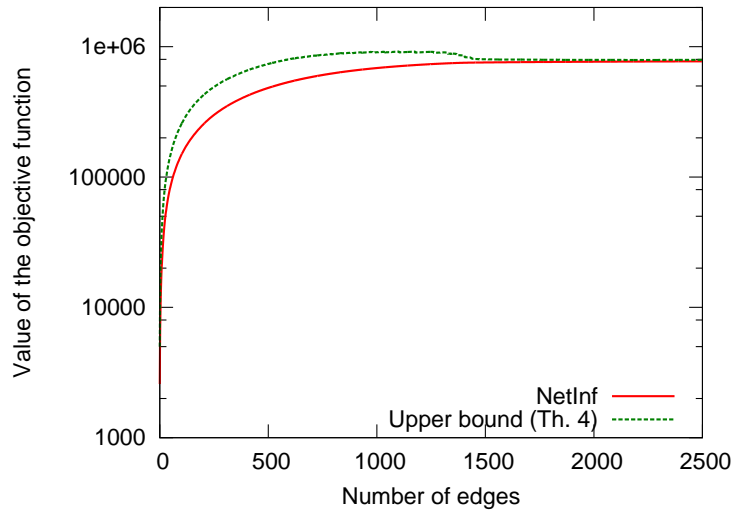} \label{fig:ScoreKROH1000NGraph}} \\
  \subfigure[Real MemeTracker data]{\includegraphics[width=0.8\textwidth]{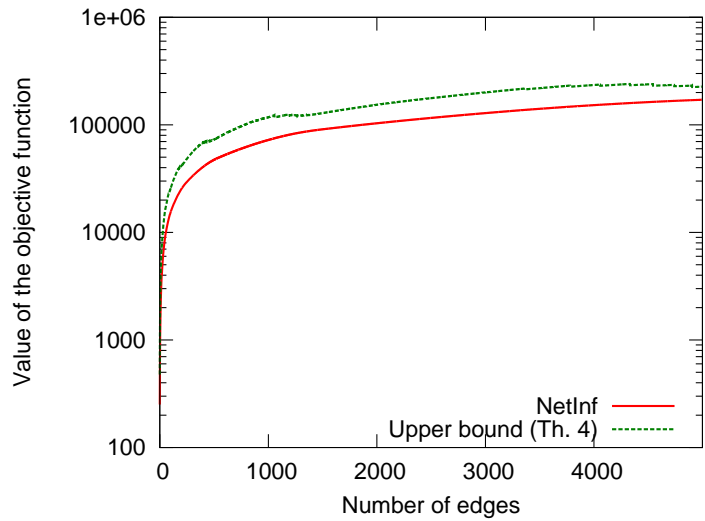} \label{fig:ScoreMemeTracker5000QOTOP1000EXPUNIF01NOTLONG100YEAR}}
  \caption{Score achieved by \netinf in comparison with the online upper bound
  from Theorem~\ref{th:online_bound}. In practice \netinf finds networks that
  are at 97\% of NP-hard to compute optimal.}
  \label{fig:Scores}
\end{figure}

\xhdr{Accuracy of \netinf} We also evaluate our approach by studying how many
edges inferred by \netinf are actually present in the true network $G^*$. We
measure the precision and recall of our method. For every value of $k$ ($1\le k
\le n(n-1)$) we generate $\hat{G_k}$ on $k$ edges by using \netinf or the
baseline method. We then compute precision (which fraction of edges in
$\hat{G_k}$ is also present $G^*$) and recall (which fraction of edges of $G^*$
appears in $\hat{G_k}$). For small $k$, we expect low recall and high precision
as we select the few edges that we are the most confident in. As $k$ increases,
precision will generally start to drop but the recall will increase.

\begin{figure*}[!h]
  \centering
  \subfigure[Hier.~Kronecker (Exp)]{\includegraphics[width=0.40\textwidth]{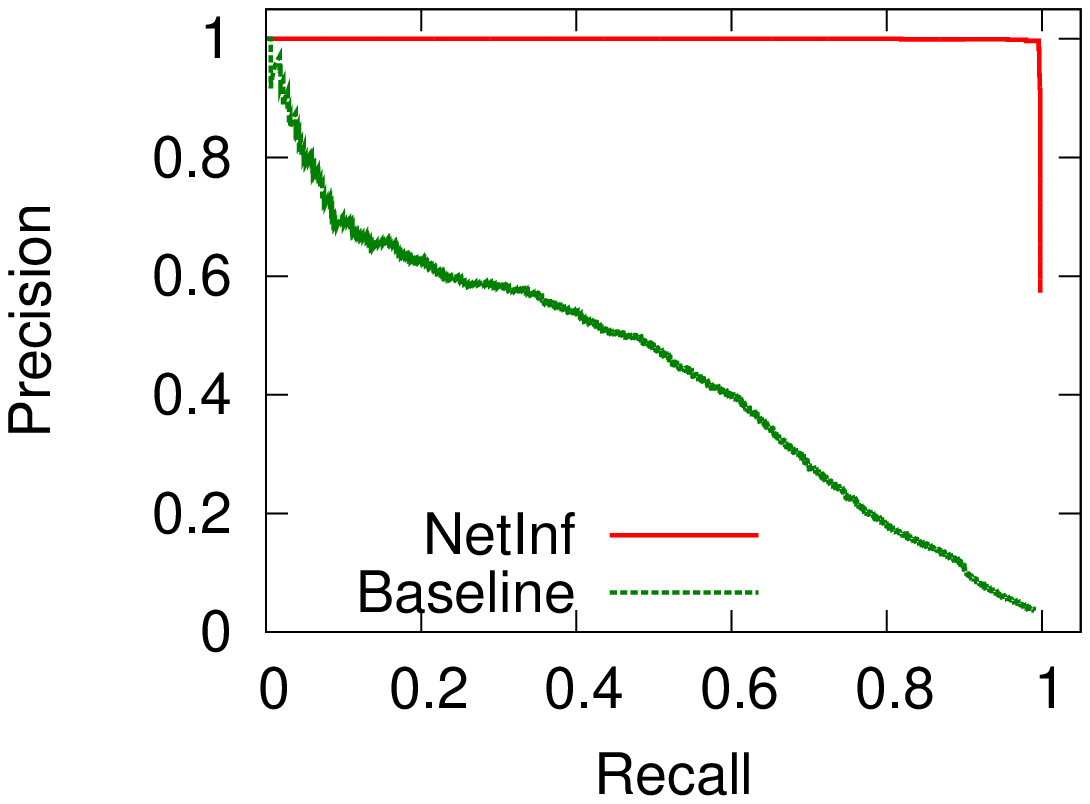} \label{fig:KROH1000NGraphExp}}
  \subfigure[Core-Periph. Kronecker (Exp)]{\includegraphics[width=0.40\textwidth]{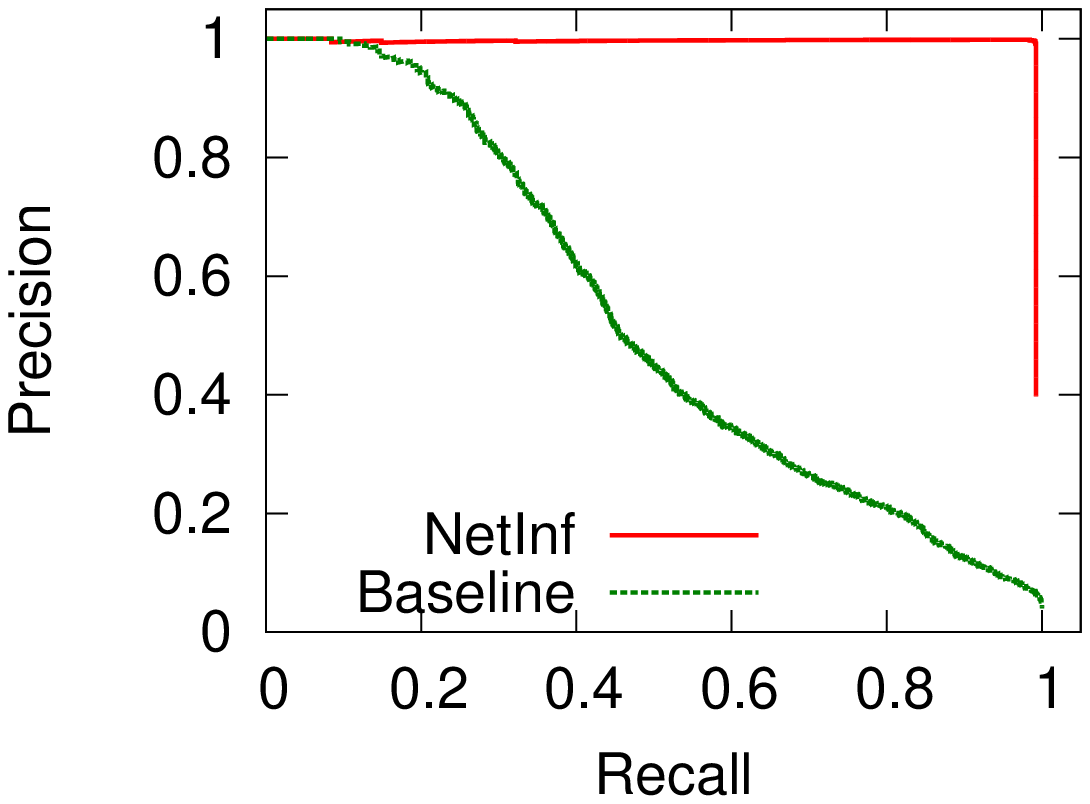} \label{fig:KROCP1000NGraphExp}}\\
  \subfigure[Flat Kronecker (Exp)]{\includegraphics[width=0.40\textwidth]{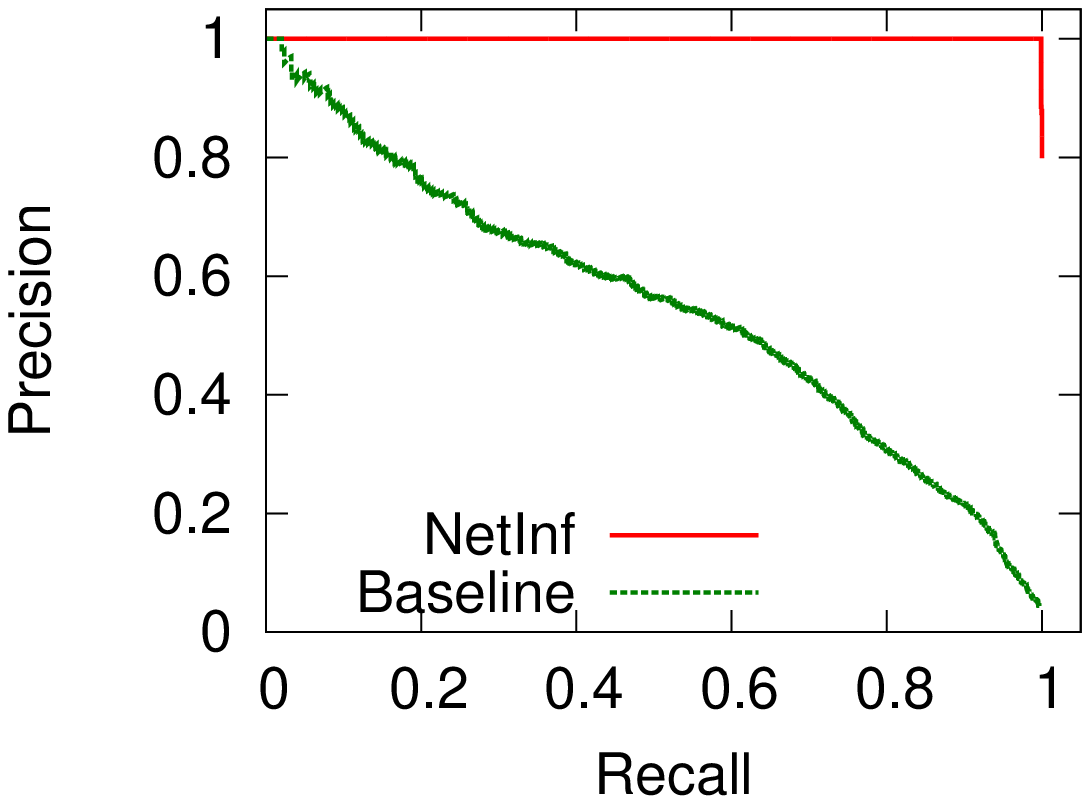} \label{fig:KROF1000NGraphExp}}
  \subfigure[Hier.~Kronecker (PL)]{\includegraphics[width=0.40\textwidth]{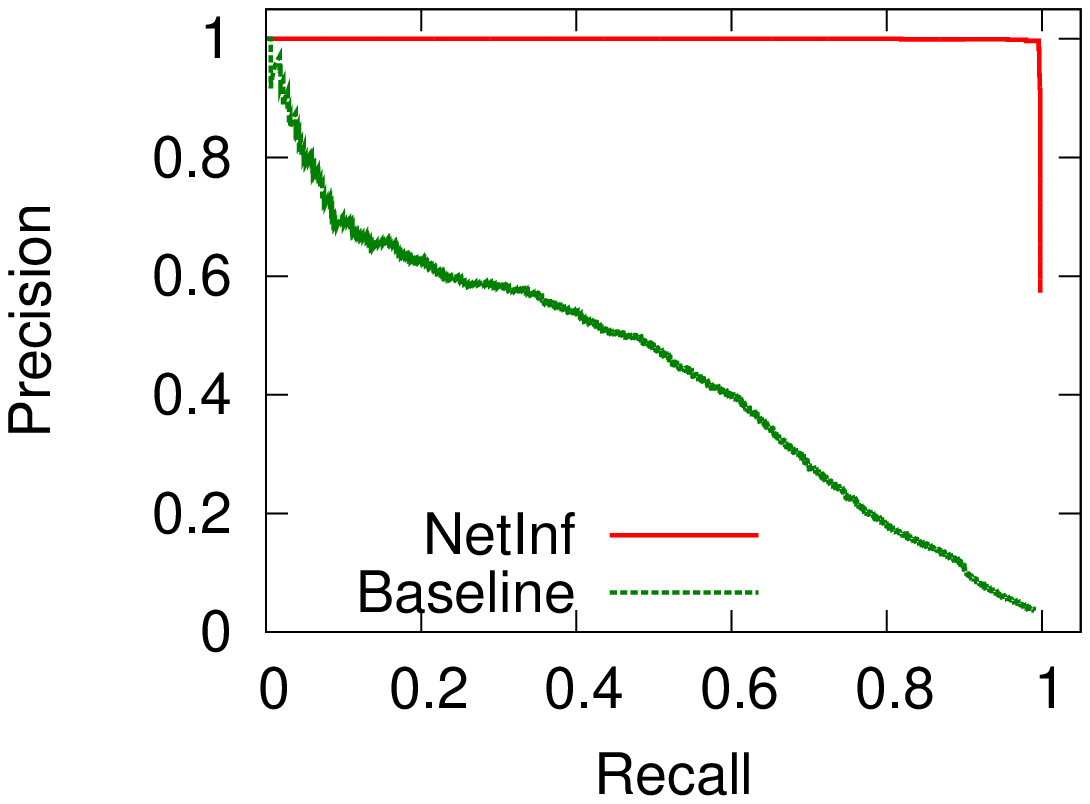}\label{fig:KROH1000NPLGraph}}\\
  \subfigure[Core-Periph. Kronecker (PL)]{\includegraphics[width=0.40\textwidth]{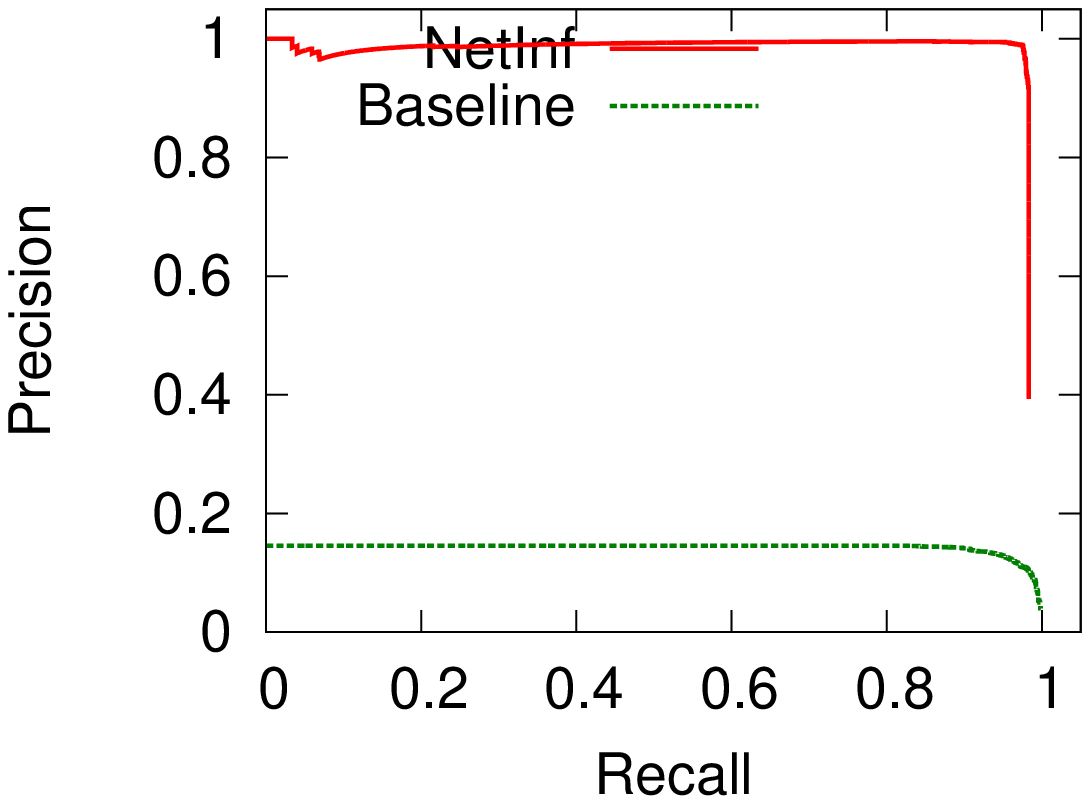}\label{fig:KROCP1000NPLGraph}}
  \subfigure[Flat Kronecker (PL)]{\includegraphics[width=0.40\textwidth]{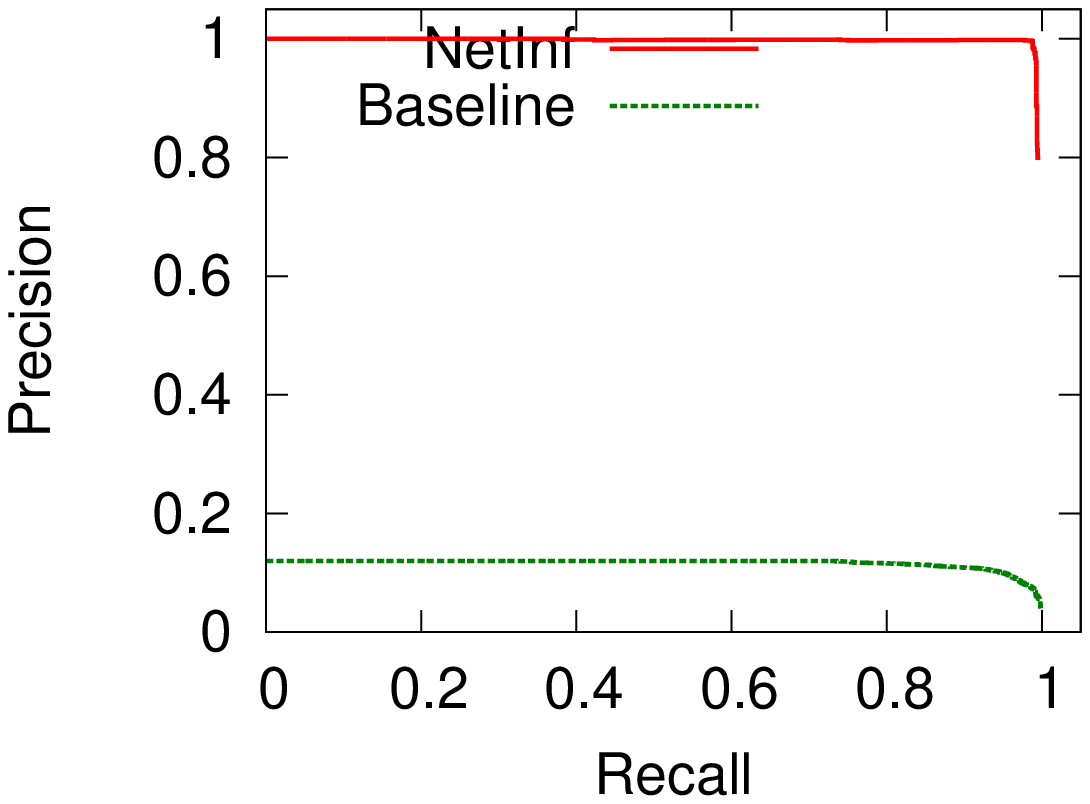}\label{fig:KROF1000NPLGraph}}\\
  \subfigure[Forest Fire (PL, $\alpha=1.1$)]{\includegraphics[width=0.40\textwidth]{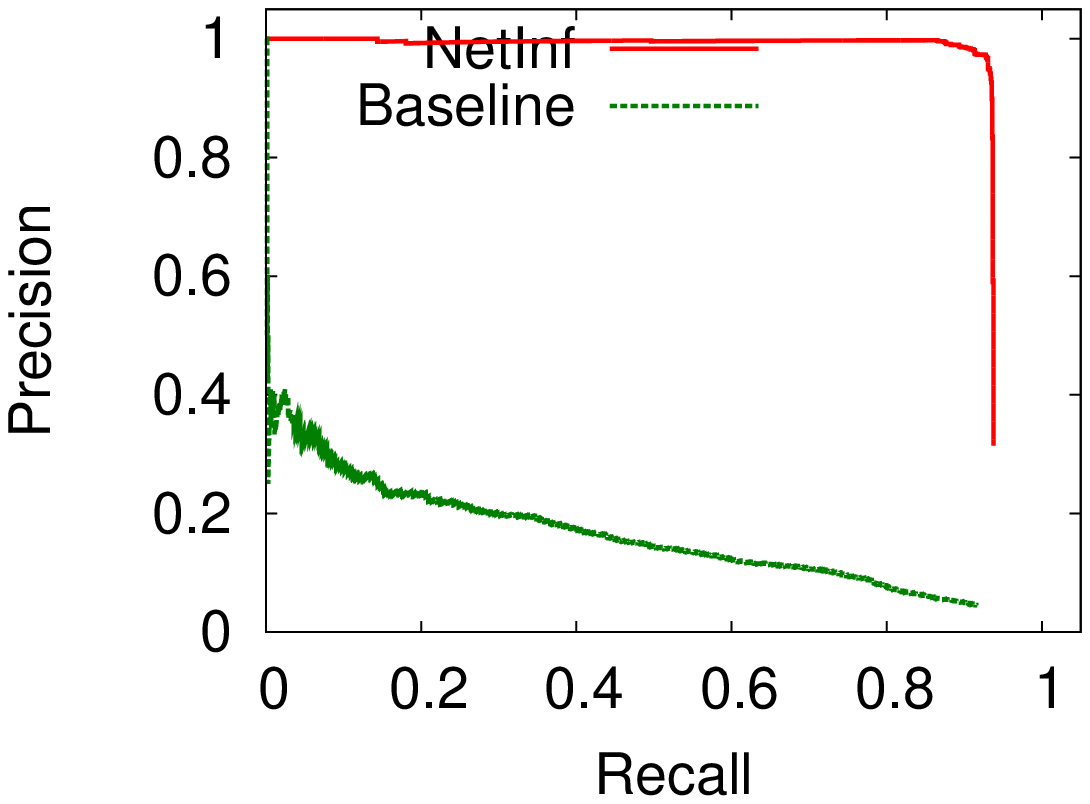}\label{fig:FF1000NGraph0}}
  \subfigure[Forest Fire (PL, $\alpha=3$)]{\includegraphics[width=0.40\textwidth]{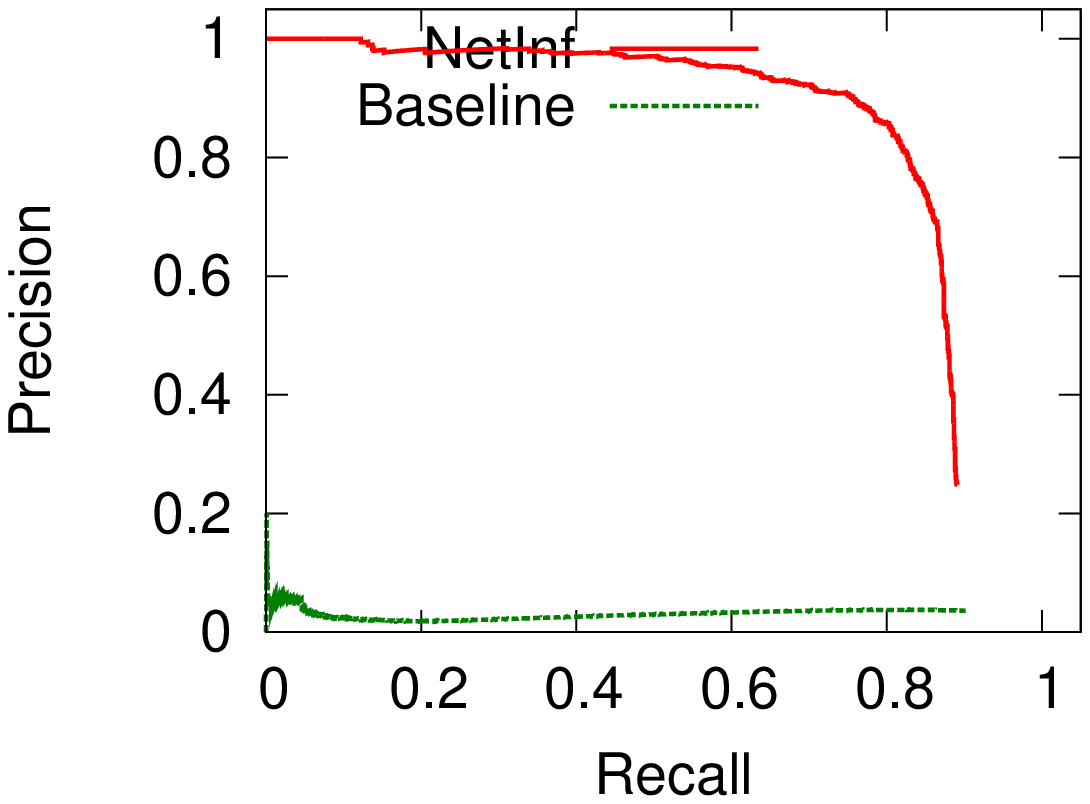}\label{fig:FF1000NGraph2}}
  \caption{Precision and recall for three 1024 node Kronecker and Forest Fire network networks
  with exponential (Exp) and power law (PL) incubation time model. The plots are generated by sweeping over values of $k$, that controls the 
  sparsity of the solution.
  } \label{fig:KRO1000NGraph}
\end{figure*}

Figure~\ref{fig:KRO1000NGraph} shows the precision-recall curves of \netinf and
the baseline method on three different Kronecker graphs (random, hierarchical
community structure and core-periphery structure) with 1024 nodes and two
incubation time models. The cascades were generated with an exponential
incubation time model with $\alpha=1$, or a power law incubation time model
with $\alpha=2$ and a value of $\beta$ low enough to avoid generating too large
cascades (in all cases, we pick a value of $\beta \in (0.1, 0.6)$). For each network 
we generated between 2,000 and 4,000 cascades so that 99\% of the edges of $G^*$ 
participated in at least one cascade. We chose cascade starting points uniformly 
at random.

First, we focus on Figures~\ref{fig:KROH1000NGraphExp},~\ref{fig:KROCP1000NGraphExp} 
and~\ref{fig:KROF1000NGraphExp} where we use the exponential incubation time model on 
different Kronecker graphs. Notice that the baseline method achieves the break-even 
point\footnote{The point at which recall is equal to precision.} between 0.4 and 0.5 
on all three networks. On the other hand, \netinf performs much better with the 
break-even point of 0.99 on all three datasets.

We view this as a particularly strong result as we were especially careful not
to generate too many cascades since more cascades mean more evidence that makes
the problem easier. Thus, using a very small number of cascades, where every
edge of $G^*$ participates in only a few cascades, we can almost perfectly
recover the underlying diffusion network $G^*$. Second important point to
notice is that the performance of \netinf seems to be strong regardless of the
structure of the network $G^*$. This means that \netinf works reliably
regardless of the particular structure of the network of which contagions
propagated (refer to Table~\ref{tab:NumCascades}).

Similarly, Figures~\ref{fig:KROH1000NPLGraph},~\ref{fig:KROCP1000NPLGraph}
and~\ref{fig:KROF1000NPLGraph} show the performance on the same three networks
but using the power law incubation time model. The performance of the baseline
now dramatically drops. This is likely due to the fact that the variance of
power-law (and heavy tailed distributions in general) is much larger than the
variance of an exponential distribution. Thus the diffusion network inference
problem is much harder in this case. As the baseline pays high price due to the
increase in variance with the break-even point dropping below $0.1$ the
performance of \netinf remains stable with the break even point in the high 90s.

\begin{figure}[t]
\centering
  \includegraphics[width=0.8\textwidth]{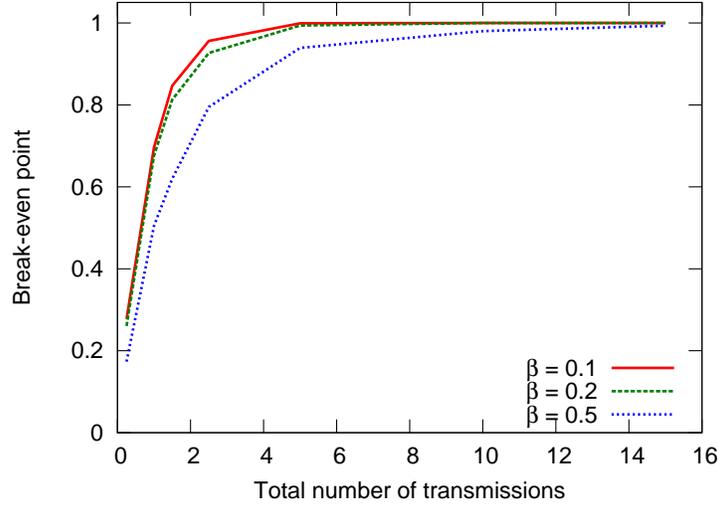}
  \caption{Performance of \netinf as a function of the amount of cascade data. The units in the x-axis are
  normalized. $x=1$ means that the total number of transmission events used for the experiment was equal 
  to the number of edges in $G^*$. On average \netinf requires about two propagation events per edge of the 
  original network in order to reliably recover the true network structure.}
  \label{fig:KRO1024CascVsPerformance}
\end{figure}

We also examine the results on the Forest Fire network (Figures~\ref{fig:FF1000NGraph0} 
and~\ref{fig:FF1000NGraph2}). Again, the performance of the baseline is very low while \netinf 
achieves the break-even point at around 0.90.

Generally, the performance on the Forest Fire network is a bit lower than on
the Kronecker graphs. However, it is important to note that while these
networks have very different global network structure (from hierarchical,
random, scale free to core periphery) the performance of \netinf is remarkably
stable and does not seem to depend on the structure of the network we are
trying to infer or the particular type of cascade incubation time model.

Finally, in all the experiments, we observe a sharp drop in precision for high values of 
recall (near $1$). This happens because the greedy algorithm starts to choose edges with 
low marginal gains that may be false edges, increasing the probability to make mistakes. 

\xhdr{Performance vs. cascade coverage} Intuitively, the larger the number of
cascades that spread over a particular edge the easier it is to identify it. On
one hand if the edge never transmitted then we can not identify it, and the
more times it participated in the transmission of a contagion the easier should
the edge be to identify.

In our experiments so far, we generated a relatively small number of cascades.
Next, we examine how the performance of \netinf depends on the amount of
available cascade data. This is important because in many real world situations
the data of only a few different cascades is available.

Figure~\ref{fig:KRO1024CascVsPerformance} plots the break-even point of \netinf
as a function of the available cascade data measured in the number of contagion
transmission events over all cascades. The total number of contagion
transmission events is simply the sum of cascade sizes. Thus, $x=1$ means that
the total number of transmission events used for the experiment was equal to
the number of edges in $G^*$. Notice that as the amount of cascade data
increases the performance of \netinf also increases. Overall we notice that
\netinf requires a total number of transmission events to be about 2 times the
number of edges in $G^*$ to successfully recover most of the edges of $G^*$.

Moreover, the plot shows the performance for different values of edge transmission
probability $\beta$. As noted before, big values of $\beta$ produce larger
cascades. Interestingly, when cascades are small (small $\beta$) \netinf needs
less data to infer the network than when cascades are larger. This occurs because 
the larger a cascade, the more difficult is to infer the parent of each node, since 
we have more potential parents for each the node to choose from.
For example, when $\beta=0.1$ \netinf needs about $2|E|$ transmission events, while 
when $\beta=0.5$ it needs twice as much data (about $4|E|$ transmissions) to obtain
the break even point of $0.9$.

\begin{figure}[t]
\centering
  \includegraphics[width=0.8\textwidth]{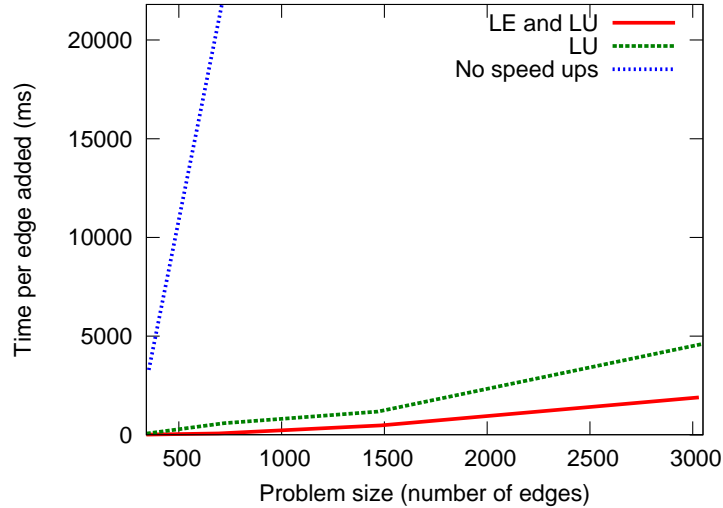}
  \caption{Average time per edge added by our algorithm implemented with lazy
  evaluation (LE) and localized update (LU).} \label{fig:KROHTime}
\end{figure}

\xhdr{Stopping criterion} In practice one does not know how long to run the
algorithm and how many edges to insert into the network $\hat{G}$. Given the
results from Figure~\ref{fig:Scores}, we found the following heuristic to give
good results. We run the \netinf algorithm for $\nedge$ steps where $\nedge$ is
chosen such that the objective function is ``close'' to the upper bound, \emph{i.e.},
$F_C(\hat{G}) > x \cdot \textrm{OPT}$, where OPT is obtained using the online
bound. In practice we use values of $x$ in range $0.8$--$0.9$. 
That means that in each iteration $k$, OPT is computed by evaluating the right hand side 
expression of the equation in Theorem~\ref{th:online_bound}, where $k$ is simply the iteration number. 
Therefore, OPT is computed online, and thus the stopping condition is also updated online. 

\xhdr{Scalability} Figure~\ref{fig:KROHTime} shows the average computation time
per edge added for the \netinf algorithm implemented with lazy evaluation and
localized update. We use a hierarchical Kronecker network and an exponential
incubation time model with $\alpha=1$ and $\beta=0.5$. Localized update speeds
up the algorithm for an order of magnitude (45$\times$) and lazy evaluation
further gives a factor of 6 improvement. Thus, overall, we achieve  two orders
of magnitude speed up (280$\times$), without {\em any} loss in solution
quality.

In practice the \netinf algorithm can easily be used to infer networks of
10,000 nodes in a matter of hours.

\begin{figure}[t]
  \centering
  \includegraphics[width=0.8\textwidth]{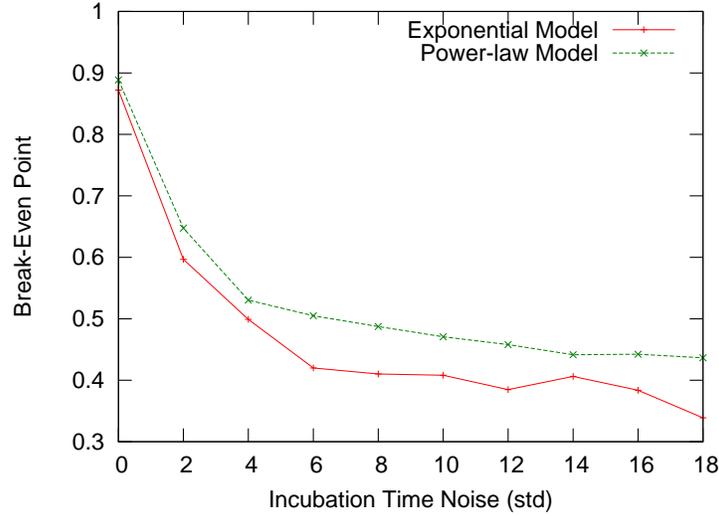}
  \caption{Break-even point of \netinf as a function of the amount of additive Gaussian noise in the incubation time.} \label{fig:NoiseWaitingTimeVsPerformance}
\end{figure}

\xhdr{Performance vs. incubation time noise} In our experiments so far, we have
assumed that the incubation time values between infections are not \emph{noisy}
and that we have access to the true distribution from which the incubation
times are drawn. However, real data may violate any of these two assumptions.

We study the performance of \netinf (break-even point) as a function of the
noise of the waiting time between infections. Thus, we add Gaussian noise to
the waiting times between infections in the cascade generation process.

Figure~\ref{fig:NoiseWaitingTimeVsPerformance} plots the performance of \netinf
(break-even point) as a function of the amount of Gaussian noise added to the
incubation times between infections for both an exponential incubation time
model with $\alpha=1$, and a power law incubation time model with $\alpha=2$.
The break-even point degrades with noise but once a high value of noise is
reached, an additional increment in the amount of noise does not degrade
further the performance of \netinf. Interestingly, the break-even point value
for high values of noise is very similar to the break-even point achieved later
in a real dataset (Figures~\ref{fig:HyperlinksGraphN500E4000C180000}
and~\ref{fig:HyperlinksMemeTrackerGraphN500C70000}).

\xhdr{Performance vs. infections by the external source} In all our experiments
so far, we have assumed that we have access to \emph{complete} cascade data,
\emph{i.e.}, we are able to observe all the nodes taking part in each cascade.
Thereby, except for the first node of a cascade, we do not have any ``jumps''
or missing nodes in the cascade as it spreads across the network. Even though
techniques for coping with missing data in information cascades have recently
been investigated~\cite{sadikov11cascades}, we evaluate \netinf against both
scenarios.

\begin{figure}[!t]
  \centering
  \subfigure[Missing node infection data]{\includegraphics[width=0.8\textwidth]{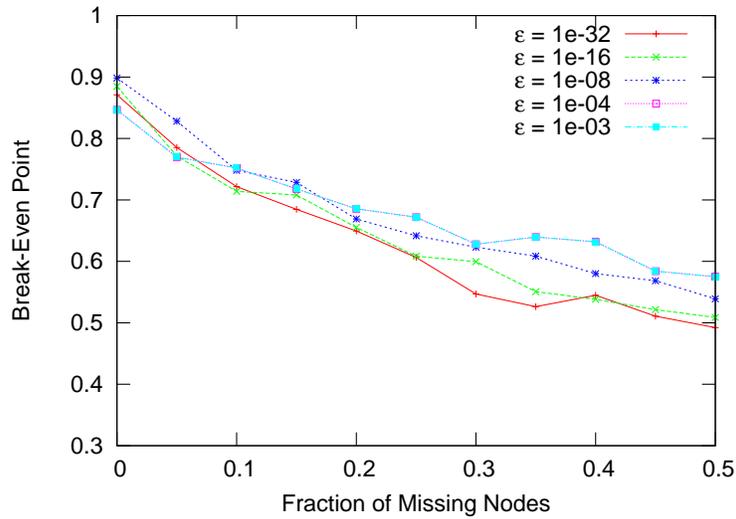} \label{fig:KroneckerMissingVsPerformance}} \\
  \subfigure[Node infections due to external source]{\includegraphics[width=0.8\textwidth]{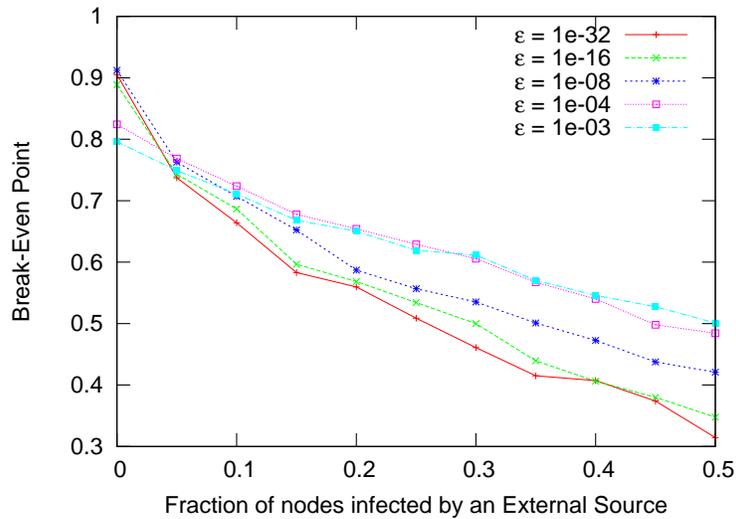} \label{fig:KroneckerExternalVsPerformance}}
  \caption{Break-even point of \netinf as (a) function of the fraction of missing nodes per cascade,
  and as (b) function of the fraction of nodes that are infected by an external source per cascade.}
  \label{fig:MissingExternalSourceVsPerformance}
\end{figure}

First, we consider the case where a random fraction of each cascade is missing. This
means that we first generate a set of cascades, but then only record node
infection times of $f$-fraction of nodes. We first generate enough cascades so
that without counting the missing nodes in the cascades, we still have that
99\% of the edges in $G^*$ participate in at least one cascade. Then we
randomly delete (\ie, set infection times to infinity) $f$-fraction of nodes in
each cascade.

Figure~\ref{fig:KroneckerMissingVsPerformance} plots the performance of \netinf
(break-even point) as a function of the percentage of missing nodes in each
cascade. Naturally, the performance drops with the amount of missing data.
However, we also note that the effect of missing nodes can be mitigated by an
appropriate choice of the parameter $\varepsilon$. Basically, higher $\varepsilon$
makes propagation via $\varepsilon$-edges more likely and thus by giving a
cascade a greater chance to propagate over the $\varepsilon$-edges \netinf can
implicitly account for the missing data.

Second, we also consider the case where the contagion does not spread through
the network via diffusion but rather due to the influence of an external
source. Thus, the contagion does not really spread over the edges of the
network but rather appears almost at random at various nodes of the network.

Figure~\ref{fig:KroneckerExternalVsPerformance} plots the performance of
\netinf (break-even point) as a function of the percentage of nodes that are
infected by an external source for different values of $\varepsilon$. In our
framework, we model the influence due to the external source with the
$\varepsilon$-edges. Note that appropriately setting $\varepsilon$ can
appropriately account for the exogenous infections that are not the result of
the network diffusion but due to the external influence. The higher the value
of $\varepsilon$, the stronger the influence of the external source,
\emph{i.e.}, we assume a greater number of missing nodes or number of nodes
that are infected by an external source. Thus, the break-even is more robust
for higher values of $\varepsilon$.
\begin{figure*}[t]
  \centering
  \includegraphics[width=0.6\textwidth]{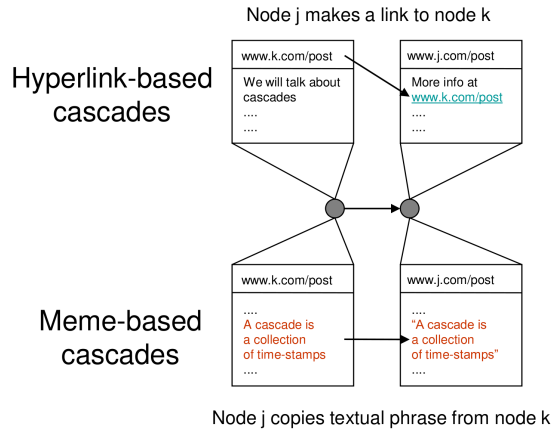}
  \caption{Hyperlink-based cascades versus meme-based cascades. In hyper-link cascades, if post $j$ linked to post $k$,
  we consider this as a contagion transmission event with the post creation time as the corresponding infection time.
  In MemeTracker cascades, we follow the spread of a short textual phrase and use post creation times as infection times.}
  \label{fig:realcascades}
\end{figure*}

\subsection{Experiments on real data}


\xhdr{Dataset description} We use more than $172$ million news articles and
blog posts from $1$ million online sources over a period of one year from
September 1 2008 till August 31 2009\footnote{Data available at
\url{http://memetracker.org} and \url{http://snap.stanford.edu/netinf}}. Based
on this raw data, we use two different methodologies to trace information on
the Web and then create two different datasets:

\noindent{\em (1) Blog hyperlink cascades dataset:} We use hyperlinks between blog 
posts to trace the flow of information~\cite{jure07cascades}. When a blog publishes a piece 
of information and uses hyper-links to refer to other posts published by other blogs we
consider this as events of information transmission. A cascade $c$ starts when
a blog publishes a post $P$ and the information propagates recursively to other
blogs by them linking to the original post or one of the other posts from which we
can trace a chain of hyperlinks all the way to the original post $P$ . By
following the chains of hyperlinks in the reverse direction we identify
hyperlink cascades~\cite{jure07cascades}. A cascade is thus composed of the
time-stamps of the hyperlink/post creation times.

\noindent{\em (1) MemeTracker dataset:} We use the MemeTracker~\cite{leskovec2009kdd} 
methodology to extract more than 343 million short textual phrases (like, ``Joe, the plumber''
or ``lipstick on a pig''). Out of these, 8 million distinct phrases appeared
more than 10 times, with the cumulative number of mentions of over 150 million.
We cluster the phrases to aggregate different textual variants of the same
phrase~\cite{leskovec2009kdd}. We then consider each phrase cluster as a
separate cascade $c$. Since all documents are time stamped, a cascade $c$ is
simply a set of time-stamps when blogs first mentioned phrase $c$. So, we
observe the times when blogs mention particular phrases but not where they
copied or obtained the phrases from. We consider the largest 5,000 cascades (phrase 
clusters) and for each website we record the time when they first mention a phrase in the particular
phrase cluster. Note that cascades in general do not spread over all the sites, which our 
methodology can successfully handle.

Figure~\ref{fig:realcascades} further illustrates the concept of hyper-link and MemeTracker 
cascades.

\begin{figure}[t]
  \centering
   \subfigure[Blog hyperlink cascades dataset]{\includegraphics[width=0.48\textwidth]{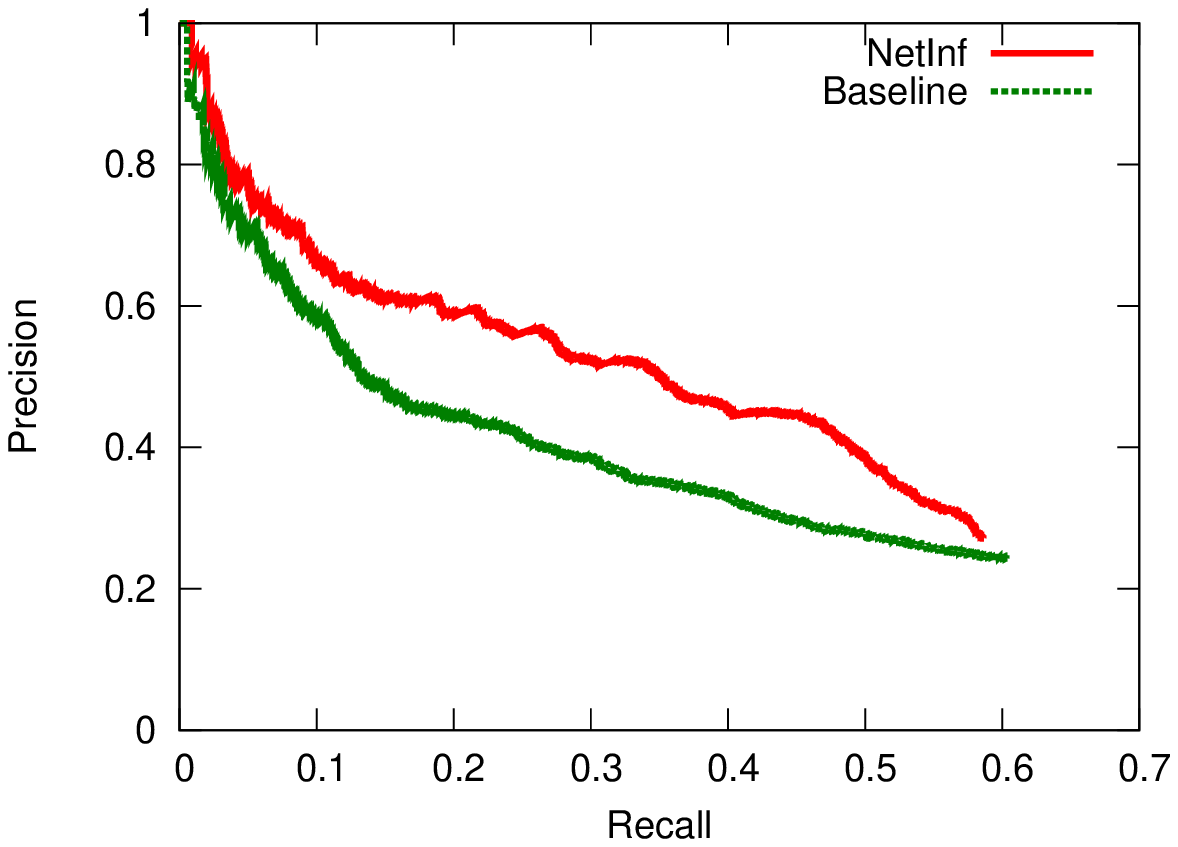}\label{fig:HyperlinksGraphN500E4000C180000}}
   \subfigure[MemeTracker dataset]{\includegraphics[width=0.48\textwidth]{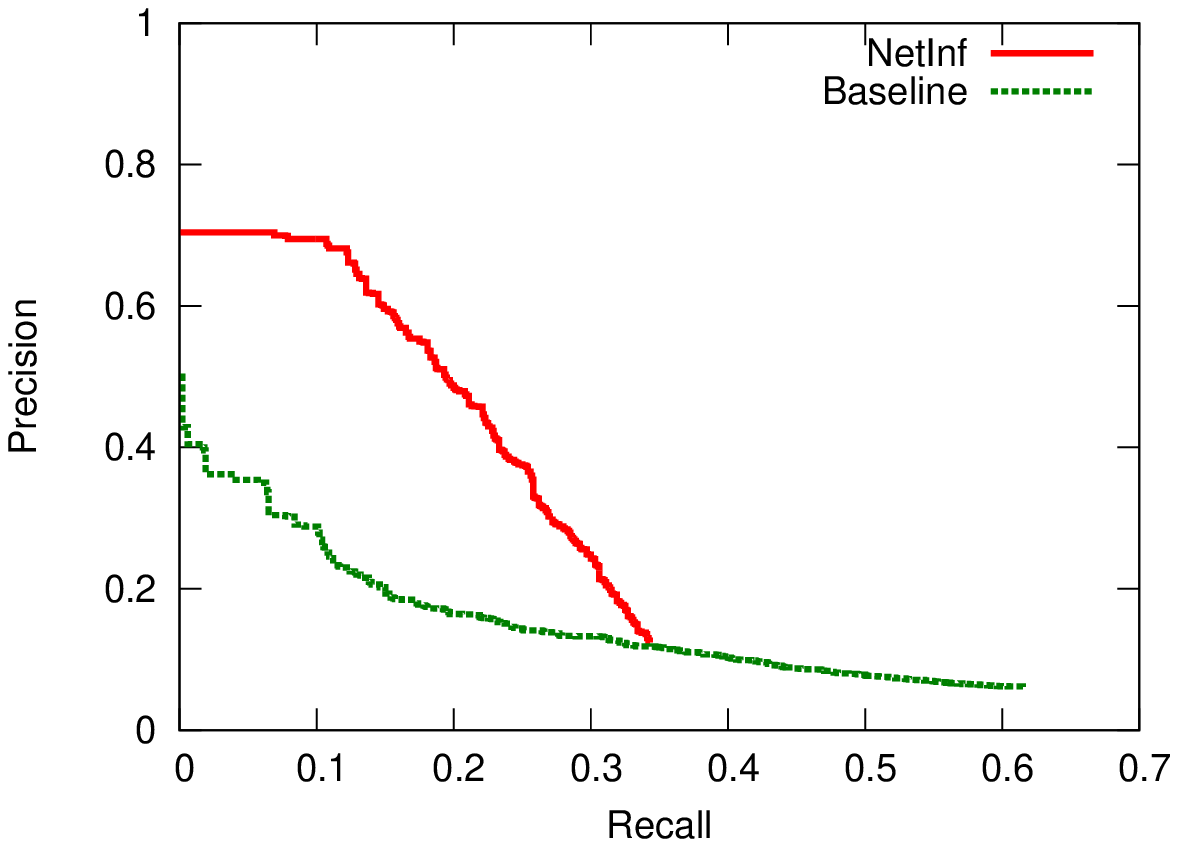}\label{fig:HyperlinksMemeTrackerGraphN500C70000}}
  \caption{Precision and recall for a 500 node hyperlink network using (a) the blog hyperlink cascades dataset (\ie, hyperlinks cascades) and (b) the MemeTracker dataset (\ie, MemeTracker cascades). We 
  used $\beta = 0.5$, $\varepsilon = 10^{-9}$ and the exponential model with $\alpha=1.0$. The time units were hours.}
\end{figure}

\xhdr{Accuracy on real data} As there is not ground truth network
for both datasets, we use the following way to create the ground
truth network $G^*$. We create a network where there is a directed edge $(u,v)$
between a pair of nodes $u$ and $v$ if a post on site $u$ linked to a post on
site $v$. To construct the network we take the top 500 sites in terms of number
of hyperlinks they create/receive. We represent each site as a node in $G^*$
and connect a pair of nodes if a post in first site linked to a post in the
second site. This process produces a ground truth network $G^*$ with 500 
nodes and 4,000 edges.

First, we use the blog hyperlink cascades dataset to infer the network $\hat{G}$ and evaluate 
how many edges $\netinf$ got right. Figure~\ref{fig:HyperlinksGraphN500E4000C180000} 
shows the performance of \netinf and the baseline. Notice that the baseline method
achieves the break-even point of 0.34, while our method performs better with a
break-even point of 0.44, almost a 30\% improvement.

\begin{figure*}[t]
  \centering
  \includegraphics[width=0.8\textwidth]{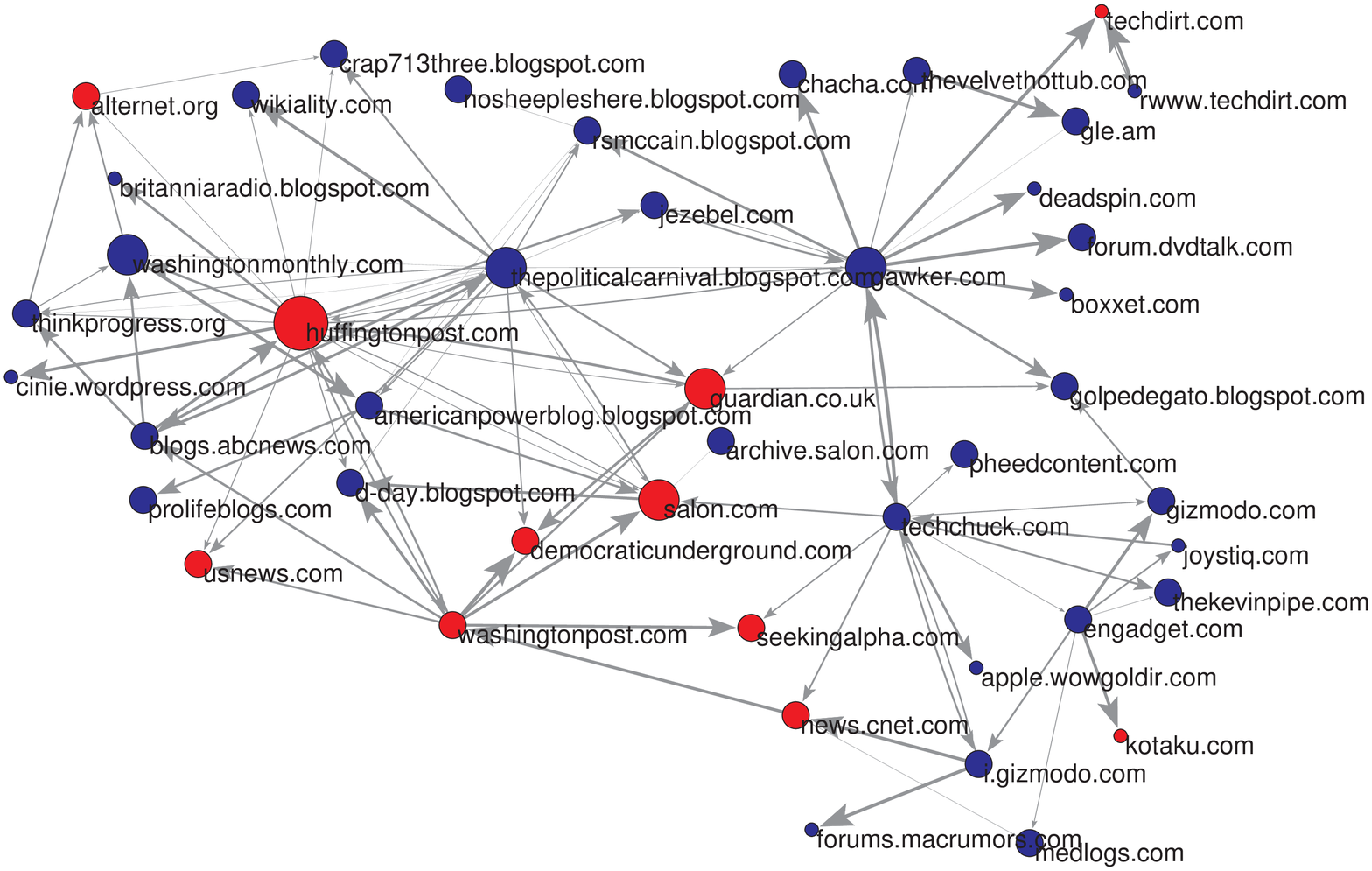}
  \caption{Small part of a news media (red) and blog (blue) diffusion network. We
  use the blog hyperlink cascades dataset, \ie, hyperlinks between blog and news media 
  posts to trace the flow of information.} \label{fig:RealNetworkBlogs}
\end{figure*}

\netinf is basically performing a link-prediction task based only on temporal linking information. The 
assumption in this experiment is that sites prefer to create links to sites that recently
mentioned information while completely ignoring the authority of the site.
Given such assumption is not satisfied in real-life, we consider the break even
point of 0.44 a good result.

Now, we consider an even harder problem, where we use the Memetracker dataset to infer 
$G^*$. In this experiment, we only observe times when sites mention particular textual phrases 
and the task is to infer the hyperlink structure of the underlying web graph.
Figure~\ref{fig:HyperlinksMemeTrackerGraphN500C70000} shows the performance of
\netinf and the baseline. The baseline method has a break-even point of 0.17
and \netinf achieves a break-even point of 0.28, more than a 50\% improvement

To have a fair comparison with the synthetic cases, notice that the exponential
incubation time model is a simplistic assumption for our real dataset, and
\netinf can potentially gain additional accuracy by choosing a more
realistic incubation time model.

\xhdr{Solution quality} Similarly as with synthetic data, in
Figure~\ref{fig:ScoreMemeTracker5000QOTOP1000EXPUNIF01NOTLONG100YEAR} we
investigate the value of the objective function and compare it to the online
bound. Notice that the bound is almost as tight as in the case of synthetic
networks, finding the solution that is least 84\% of optimal and both curves
are similar in shape to the synthetic case value. Again, as in the synthetic
case, the value of the objective function quickly flattens out which means that
one needs a relatively few number of edges to capture most of the information
flow on the Web.

\begin{figure*}[t]
\centering
\includegraphics[width=\textwidth]{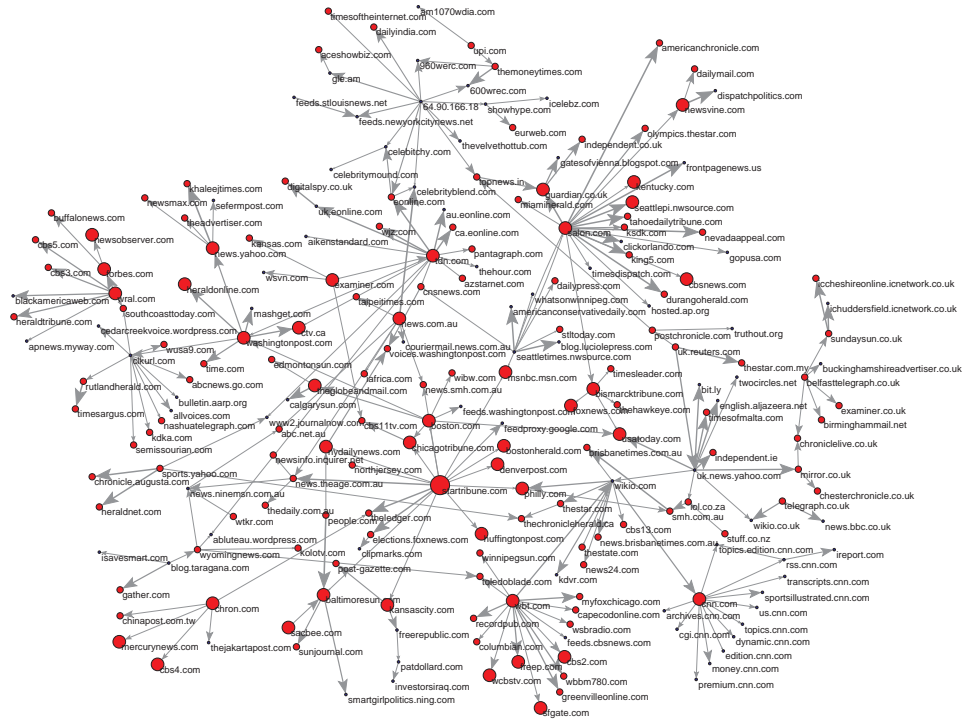}
\caption{Small part of a news media (red) and blog (blue) diffusion network. We use the MemeTracker dataset, \ie, textual 
phrases from MemeTracker to trace the flow of information.} \label{fig:RealNetworkBlogsMemes}
\end{figure*}

In the remainder of the section, we use the top 1,000 media sites and blogs with the largest number of documents.

\xhdr{Visualization of diffusion networks}
We examine the structure of the inferred diffusion networks using both datasets: the
blog hyperlink cascades dataset and the MemeTracker dataset.

Figure~\ref{fig:RealNetworkBlogs} shows the largest connected component of the
diffusion network after $100$ edges have been chosen using the first
dataset, \emph{i.e.}, using hyperlinks to track the flow of information.
The size of the nodes is proportional to the number of articles on the site and
the width of the edge is proportional to the probability of influence,
\emph{i.e.}, stronger edges have higher width. The strength of an edge across all 
cascades is simply defined as the marginal gain given by adding the edge in the 
greedy algorithm (and this is proportional to the probability of influence). Since news 
media articles rarely use hyperlinks to refer to one another, the network is somewhat 
biased towards web blogs (blue nodes). There are several interesting patterns to 
observe.

First, notice that three main clusters emerge: on the left side of the network
we can see blogs and news media sites related to politics, at the right top, we
have blogs devoted to gossip, celebrity news or entertainment and on the right
bottom, we can distinguish blogs and news media sites that deal with
technological news. As Huffington Post and Political Carnival play the central
role on the political side of the network, mainstream media sites like
Washington Post, Guardian and the professional blog Salon.com play the role of
connectors between the different parts of the network. The celebrity gossip
part of the network is dominated by the blog Gawker and technology news gather
around blogs Gizmodo and Engadget, with CNet and TechChuck establishing the
connection to the rest of the network. 

Figure~\ref{fig:RealNetworkBlogsMemes} shows the largest connected component of 
the diffusion network after $300$ edges have been chosen using the second methodology, 
\emph{i.e.} using short textual phrases to track the flow of information. In this case, the network 
is biased towards news media sites due to its higher volume of information.

\xhdr{Insights into the diffusion on the web} The inferred diffusion networks
also allow for analysis of the global structure of information propagation on
the Web. For this analysis, we use the MemeTracker dataset and analyze the structure 
of the inferred information diffusion network.

\begin{figure*}[t]
  \centering
  \subfigure[Influence Index]{\includegraphics[width=0.48\textwidth]{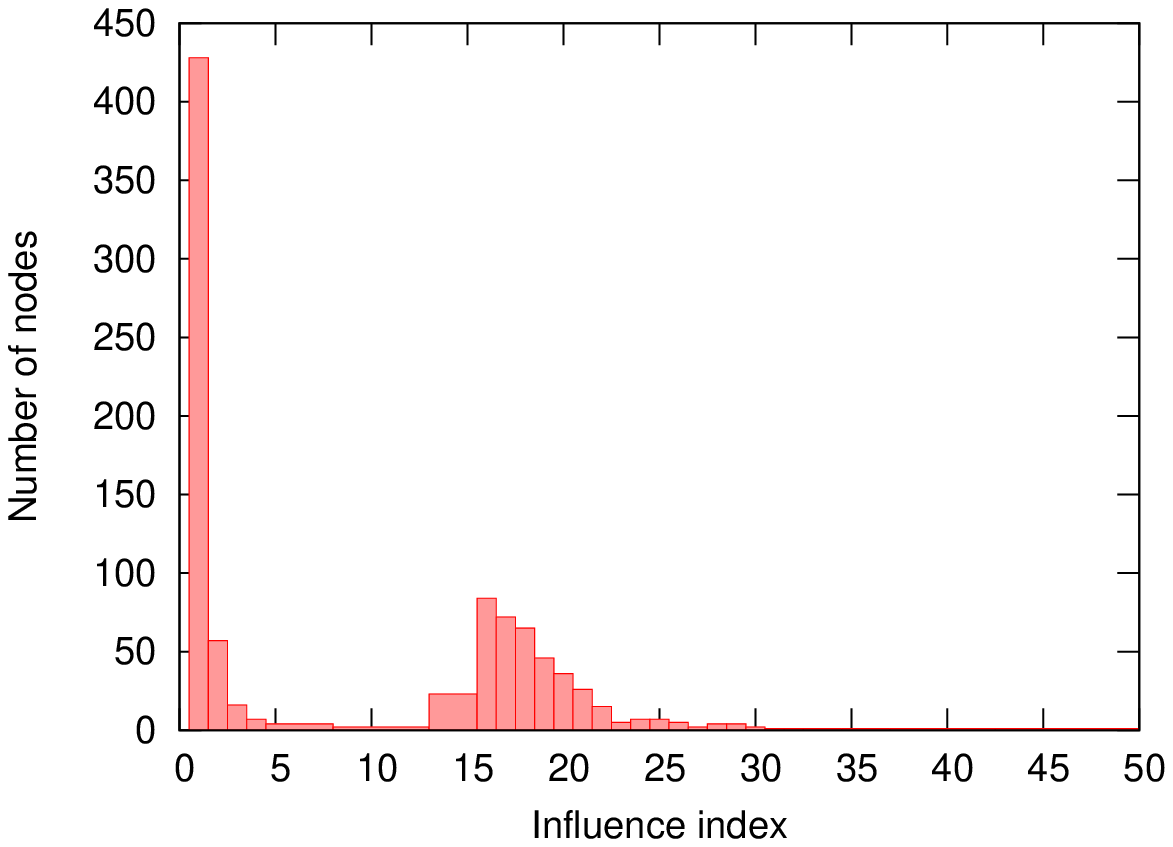} \label{fig:InfluenceIndex5000QOTOP1000EXPUNIF01NOTLONG100YEARP}}
  \subfigure[Number of edges as iterations proceed]{\includegraphics[width=0.48\textwidth]{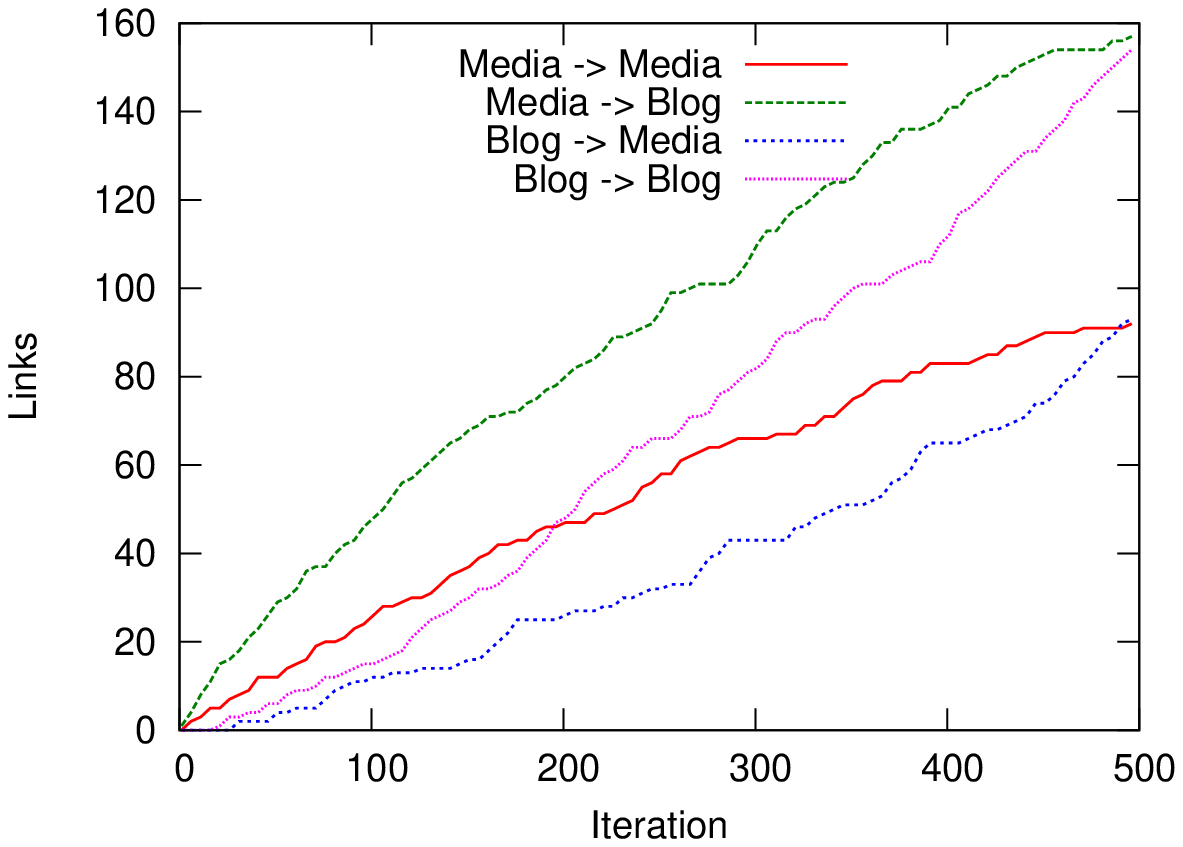} \label{fig:links}} \\
  \subfigure[Median edge time lag]{\includegraphics[width=0.8\textwidth]{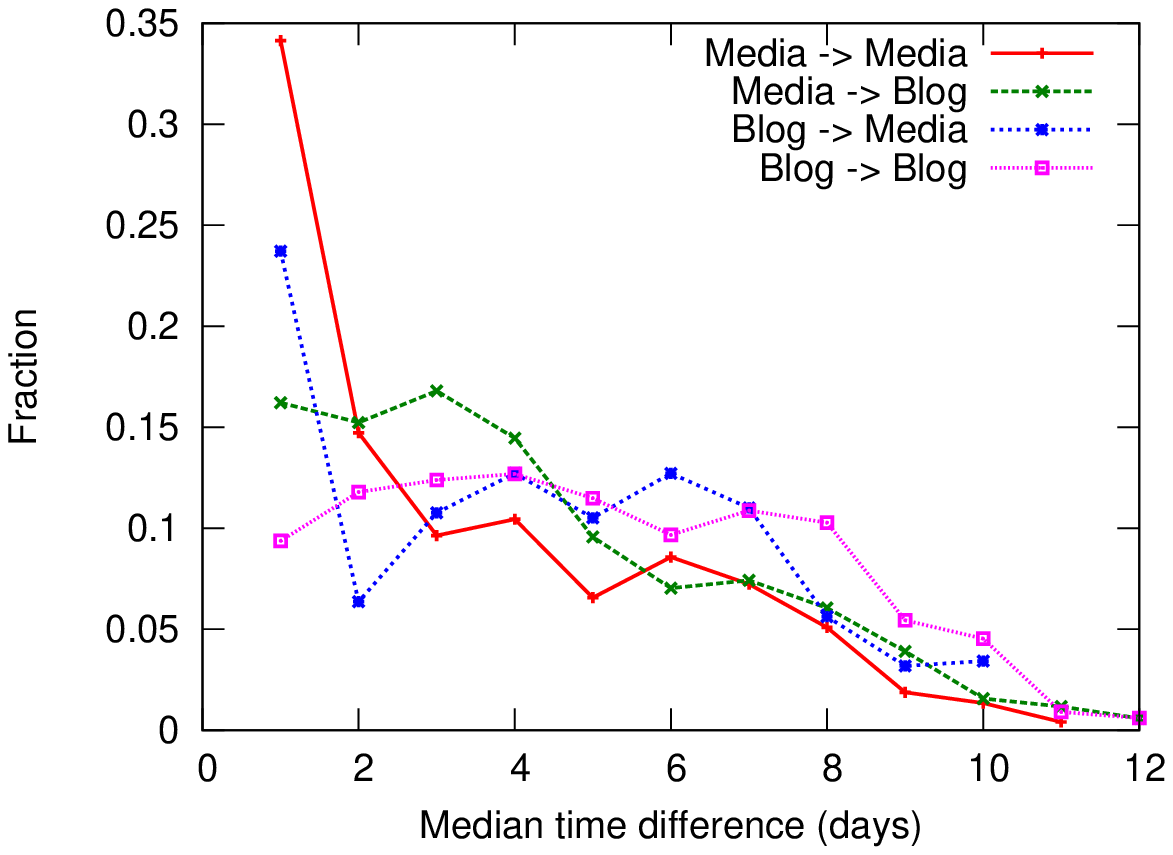} \label{fig:MedianTimeQ5000}}
  \caption{(a) Distribution of node influence index. Most nodes have very low influence (they act as sinks).
  (b) Number and strength of edges between different media types. Edges of news media influencing blogs are the strongest.
  (c) Median time lag on edges of different type.}
\end{figure*}

First, Figure~\ref{fig:InfluenceIndex5000QOTOP1000EXPUNIF01NOTLONG100YEARP}
shows the distribution of the influence index. The influence index is defined as the number 
of reachable nodes from $w$ by traversing edges of the inferred diffusion network
(while respecting edge directions). Nevertheless, we are also interested in the
distance from $w$ to its reachable nodes, i.e. nodes at shorter distances are
more likely to be infected by $w$. Thus, we slightly modify the definition of
influence index to be $\sum_{u} 1/d_{wu}$ where we sum over all the reachable
nodes from $w$ and $d_{wu}$ is the distance between $w$ and $u$.
Notice that we have two types of nodes. There is a small set of nodes that can reach many 
other nodes, which means they either directly or indirectly propagate information to them. On 
the other side we have a large number of sites that only get influenced but do not influence 
many other sites. This hints at a core periphery structure of the diffusion network
with a small set of sites directly or indirectly spreading the information in the rest of the 
network.

Figure~\ref{fig:links} investigates the number of links in the inferred network
that point between different types of sites. Here we split the sites into
mainstream media and blogs. Notice that most of the links point from news media
to blogs, which says that most of the time information propagates from the
mainstream media to blogs. Then notice how at first many media-to-media links are
chosen but in later iterations the increase of these links starts to slow down.
This means that media-to-media links tend to be the strongest and \netinf picks
them early. The opposite seems to occur in case of blog-to-blog links where
relatively few are chosen first but later the algorithm picks more of them.
Lastly, links capturing the influence of blogs on mainstream media are the
rarest and weakest. This suggests that most information travels from mass media
to blogs.

Last, Figure~\ref{fig:MedianTimeQ5000} shows the median time difference between
mentions of different types of sites. For every edge of the inferred diffusion
network, we compute the median time needed for the information to spread from
the source to the destination node. Again, we distinguish the mainstream media
sites and blogs. Notice that media sites are quick to infect one another or
even to get infected from blogs. However, blogs tend to be much slower in
propagating information. It takes a relatively long time for them to get
``infected'' with information regardless whether the information comes from the
mainstream media or the blogosphere.

Finally, we have observed that the insights into diffusion on the web using the 
inferred network are very similar to insights obtained by simply taking the hyperlink network. However, 
our aim here is to show that (i) although the quantitative results are modest in terms of precision and recall, 
the qualitative insights makes sense, and that (ii) it is surprising that using simply timestamps of links, we 
are able to draw the same qualitative insights as using the hyperlink network

\section{Further related work}
\label{sec:related}
There are several lines of work we build upon. Although the information diffusion in on-line settings has received considerable
attention~\cite{gruhl2004information,kumar2004structure,adar05epidemics,jure06viral,jure06influence,jure07cascades,nowell08letter},
only a few studies were able to study the actual shapes of cascades~\cite{jure07cascades,nowell08letter,ghosh2011framework,romero2011differences,ver2011stops}. 
The problem of inferring links of diffusion was first studied by Adar and Adamic~\cite{adar05epidemics}, who formulated it as a supervised classification 
problem and used Support Vector Machines combined with rich textual features to predict the occurrence of individual links. Although rich textual features 
are used, links are predicted independently and thus their approach is similar to our baseline method in the sense that it picks a threshold 
(\emph{i.e.}, hyperplane in case of SVMs) and predicts in\-di\-vi\-dua\-lly the most probable links.

The work most closely related to our approach, \connie~\cite{meyers10netinf} and \netrate~\cite{manuel11icml}, also uses a generative probabilistic model for the 
problem of inferring a latent social network from diffusion (cascades) data. However, \connie and \netrate use convex programming to solve the network inference 
problem. \connie includes a $l_1$-like penalty term that controls sparsity while \netrate provides a unique sparse solution by allowing different transmission rates 
across edges. For each edge $(i, j)$, \connie infers a prior probability $\beta_{i,j}$ and \netrate infers a transmission rate $\alpha_{i,j}$. Both algorithms are 
computationally more expensive than \netinf.
In our work, we assume that all edges of the network have the same prior probability ($\beta$) and transmission rate ($\alpha$). From this point of view, we think the 
comparison between the algorithms is unfair since \netrate and \connie have more degrees of freedom

Network structure learning has been considered for estimating the dependency structure of probabilistic graphical models~\cite{friedman2003being,friedman1999learning}.
However, there are fundamental differences between our approach and graphical models structure learning. 
(a) we learning directed networks, but Bayes netws are DAGs
(b) undirected graphical model structure learning makes no assumption
about the network but they learn undirected and we learn directed
networks

First, our work makes no assumption about the network structure (we allow cycles, reciprocal edges) and are thus able to learn general directed networks. In directed graphical models, 
reciprocal edges and cycles are not allowed, and the inferred network is a directed acyclic graph (DAG). In undirected graphical models, there are typically no assumptions about 
the network structure, but the inferred network is undirected. Second, Bayesian network structure inference methods are generally heuristic approaches without any approximation guarantees.
Network structure learning has also been used for estimating epidemiological networks~\cite{wallinga04epidemic} and for estimating probabilistic relational 
models~\cite{getoor2003learning}. In both cases, the problem is formulated in a probabilistic framework. However, since the problem is intractable, heuristic greedy 
hill-climbing or stochastic search that offer no performance guarantee were usually used in practice. 
In contrast, our work provides a novel formulation and a {\em tractable} solution together with an approximation guarantee.

Our work relates to static sparse graph estimation using graphical Lasso methods~\cite{wainwright06graphical,schmidt2007learning,friedman08lasso,meinshausen2006high},
unsupervised structure network inference using kernel methods~\cite{lippert2008kernel}, mutual information relevance network inference~\cite{butte2000mutual}, inference of influence 
probabilities~\cite{goyal2010learning}, and extensions to time evolving graphical models~\cite{ahmed2009recovering,ghahramani1998learning,song09timevarying}. Our work 
is also related to a link prediction problem~\cite{janse03linkpred,taskar03linkpred,libennowell03linkpred,backstrom11srw,vert2005supervised} but different in a sense that this line 
of work assumes that part of the network is already visible to us.

Last, although \emph{submodular} function maximization has been previously considered for sensor placement~\cite{leskovec2007cost} and finding influencers
in viral marketing~\cite{kempe03maximizing}, to the best of our knowledge, the present work is the first that considers submodular function maximization in
the context of network structure learning.

\section{Conclusions}
\label{sec:conclusions}
We have investigated the problem of tracing paths of diffusion and influence.
We formalized the problem and developed a scalable algorithm, \netinf, to infer
networks of influence and diffusion.  First, we defined a generative model of
cascades and showed that choosing the best set of $\nedge$ edges maximizing the
likelihood of the data is NP-hard. By exploiting the submodularity of our
objective function, we developed \netinf, an efficient algorithm for inferring
a near-optimal set of $\nedge$ directed edges. By exploiting localized updates
and lazy evaluation, our algorithm is able to scale to very large real data
sets.

We evaluated our algorithm on synthetic cascades sampled from our generative
model, and showed that \netinf is able to accurately recover the underlying
network from a relatively small number of samples.  In our experiments, \netinf
drastically outperformed a naive maximum weight baseline heuristic.

Most importantly, our algorithm allows us to study properties of real networks.
We evaluated \netinf on a large real data set of memes propagating across news
websites and blogs. We found that the inferred network exhibits a
core-periphery structure with mass media influencing most of the blogosphere.
Clusters of sites related to similar topics emerge (politics, gossip,
technology, etc.), and a few sites with social capital interconnect these
clusters allowing a potential diffusion of information among sites in different
clusters.

There are several interesting directions for future work. Here we only used
time difference to infer edges and thus it would be interesting to utilize more
informative features (e.g., textual content of postings etc.) to more
accurately estimate the influence probabilities. Moreover, our  work considers
static propagation networks, however real influence networks are dynamic and
thus it would be interesting to relax this assumption. Last, there are many
other domains where our methodology could be useful: inferring interaction
networks in systems biology (protein-protein and gene interaction networks),
neuroscience (inferring physical connections between neurons) and epidemiology.

We believe that our results provide a promising step towards understanding
complex processes on networks based on partial observations.

\subsection*{Acknowledgments}
We thank Spinn3r for resources that facilitated the research. The research was
supported in part by Albert Yu \& Mary Bechmann Foundation, IBM, Lightspeed,
Microsoft, Yahoo, grants ONR N00014-09-1-1044, NSF CNS0932392, NSF CNS1010921,
NSF IIS1016909, NSF IIS0953413, AFRL FA8650-10-C-7058 and Okawa Foundation
Research Grant. Ma\-nu\-el Go\-mez Ro\-dri\-guez has been supported in part by a
Fundacion Caja Madrid Graduate Fe\-llow\-ship, a Fundacion Barrie de la Maza
Graduate Fe\-llow\-ship and by the Max Planck So\-ciety.

\bibliographystyle{acmtrans}
\bibliography{refs}

\begin{thebibliography}{}

\bibitem[\protect\citeauthoryear{Adar and Adamic}{Adar and
  Adamic}{2005}]{adar05epidemics}
{\sc Adar, E.} {\sc and} {\sc Adamic, L.~A.} 2005.
\newblock Tracking information epidemics in blogspace.
\newblock In {\em Web Intelligence}. 207--214.

\bibitem[\protect\citeauthoryear{Adar, Zhang, Adamic, and Lukose}{Adar
  et~al\mbox{.}}{2004}]{adar04blogspace}
{\sc Adar, E.}, {\sc Zhang, L.}, {\sc Adamic, L.~A.}, {\sc and} {\sc Lukose,
  R.~M.} 2004.
\newblock Implicit structure and the dynamics of blogspace.
\newblock In {\em Workshop on the Weblogging Ecosystem}.

\bibitem[\protect\citeauthoryear{Ahmed and Xing}{Ahmed and
  Xing}{2009}]{ahmed2009recovering}
{\sc Ahmed, A.} {\sc and} {\sc Xing, E.} 2009.
\newblock {Recovering time-varying networks of dependencies in social and
  biological studies}.
\newblock In {\em PNAS '09: Proceedings of the National Academy of Sciences}.
  Vol. 106.

\bibitem[\protect\citeauthoryear{Anderson and May}{Anderson and
  May}{2002}]{anderson92infectious}
{\sc Anderson, R.~M.} {\sc and} {\sc May, R.~M.} 2002.
\newblock {\em Infectious diseases of humans: Dynamics and control}.
\newblock Oxford Press.

\bibitem[\protect\citeauthoryear{Backstrom and Leskovec}{Backstrom and
  Leskovec}{2011}]{backstrom11srw}
{\sc Backstrom, L.} {\sc and} {\sc Leskovec, J.} 2011.
\newblock Supervised random walks: Predicting and recommending links in social
  networks.
\newblock In {\em WSDM '11: Proceedings of the ACM International Conference on
  Web Search and Data Mining}.

\bibitem[\protect\citeauthoryear{Bailey}{Bailey}{1975}]{bailey75mathematical}
{\sc Bailey, N. T.~J.} 1975.
\newblock {\em The Mathematical Theory of Infectious Diseases and its
  Applications\/}, 2nd ed.
\newblock Hafner Press.

\bibitem[\protect\citeauthoryear{Barab\'{a}si}{Barab\'{a}si}{2005}]{barabasi05human}
{\sc Barab\'{a}si, A.-L.} 2005.
\newblock The origin of bursts and heavy tails in human dynamics.
\newblock {\em Nature\/}~{\em 435}, 207.

\bibitem[\protect\citeauthoryear{Barab\'{a}si and Albert}{Barab\'{a}si and
  Albert}{1999}]{barabasi99emergence}
{\sc Barab\'{a}si, A.-L.} {\sc and} {\sc Albert, R.} 1999.
\newblock Emergence of scaling in random networks.
\newblock {\em Science\/}~{\em 286}, 509--512.

\bibitem[\protect\citeauthoryear{Butte and Kohane}{Butte and
  Kohane}{2000}]{butte2000mutual}
{\sc Butte, A.} {\sc and} {\sc Kohane, I.} 2000.
\newblock Mutual information relevance networks: functional genomic clustering
  using pairwise entropy measurements.
\newblock In {\em Pac Symp Biocomput}. Vol.~5. 418--429.

\bibitem[\protect\citeauthoryear{Clauset, Moore, and Newman}{Clauset
  et~al\mbox{.}}{2008}]{clauset08hierarchical}
{\sc Clauset, A.}, {\sc Moore, C.}, {\sc and} {\sc Newman, M. E.~J.} 2008.
\newblock Hierarchical structure and the prediction of missing links in
  networks.
\newblock {\em Nature\/}~{\em 453,\/}~7191, 98--101.

\bibitem[\protect\citeauthoryear{Crane and Sornette}{Crane and
  Sornette}{2008}]{crane08response}
{\sc Crane, R.} {\sc and} {\sc Sornette, D.} 2008.
\newblock Robust dynamic classes revealed by measuring the response function of
  a social system.
\newblock {\em PNAS '08: Proceedings of the National Academy of
  Sciences\/}~{\em 105,\/}~41 (October), 15649--15653.

\bibitem[\protect\citeauthoryear{Domingos and Richardson}{Domingos and
  Richardson}{2001}]{domingos01mining}
{\sc Domingos, P.} {\sc and} {\sc Richardson, M.} 2001.
\newblock Mining the network value of customers.
\newblock In {\em KDD '01: Proceedings of the 7th ACM SIGKDD international
  conference on Knowledge discovery and data mining}.

\bibitem[\protect\citeauthoryear{Edmonds}{Edmonds}{1967}]{edmonds67branchings}
{\sc Edmonds, J.} 1967.
\newblock Optimum branchings.
\newblock {\em Journal of Research of the National Bureau of Standards\/}~71B,
  233--240.

\bibitem[\protect\citeauthoryear{Erd\H{o}s and R\'{e}nyi}{Erd\H{o}s and
  R\'{e}nyi}{1960}]{erdos60random}
{\sc Erd\H{o}s, P.} {\sc and} {\sc R\'{e}nyi, A.} 1960.
\newblock On the evolution of random graphs.
\newblock {\em Publication of the Mathematical Institute of the Hungarian
  Academy of Science\/}~{\em 5}, 17--67.

\bibitem[\protect\citeauthoryear{Friedman, Hastie, and Tibshirani}{Friedman
  et~al\mbox{.}}{2008}]{friedman08lasso}
{\sc Friedman, J.}, {\sc Hastie, T.}, {\sc and} {\sc Tibshirani, R.} 2008.
\newblock {Sparse inverse covariance estimation with the graphical lasso}.
\newblock {\em Biostat\/}~{\em 9,\/}~3, 432--441.

\bibitem[\protect\citeauthoryear{Friedman and Koller}{Friedman and
  Koller}{2003}]{friedman2003being}
{\sc Friedman, N.} {\sc and} {\sc Koller, D.} 2003.
\newblock {Being Bayesian about network structure. A Bayesian approach to
  structure discovery in Bayesian networks}.
\newblock {\em Machine Learning\/}~{\em 50,\/}~1, 95--125.

\bibitem[\protect\citeauthoryear{Friedman, Nachman, and Pe'er}{Friedman
  et~al\mbox{.}}{1999}]{friedman1999learning}
{\sc Friedman, N.}, {\sc Nachman, I.}, {\sc and} {\sc Pe'er, D.} 1999.
\newblock {Learning Bayesian network structure from massive datasets: The
  ``Sparse Candidate'' algorithm}.
\newblock In {\em UAI '99: Proceedings of the 15th Conference on Uncertainty in
  Artificial Intelligence}.

\bibitem[\protect\citeauthoryear{Getoor, Friedman, Koller, and Taskar}{Getoor
  et~al\mbox{.}}{2003}]{getoor2003learning}
{\sc Getoor, L.}, {\sc Friedman, N.}, {\sc Koller, D.}, {\sc and} {\sc Taskar,
  B.} 2003.
\newblock {Learning probabilistic models of link structure}.
\newblock {\em The Journal of Machine Learning Research\/}~{\em 3}, 707.

\bibitem[\protect\citeauthoryear{Ghahramani}{Ghahramani}{1998}]{ghahramani1998learning}
{\sc Ghahramani, Z.} 1998.
\newblock {Learning dynamic Bayesian networks}.
\newblock In {\em Adaptive Processing of Sequences and Data Structures}.

\bibitem[\protect\citeauthoryear{Ghosh and Lerman}{Ghosh and
  Lerman}{2011}]{ghosh2011framework}
{\sc Ghosh, R.} {\sc and} {\sc Lerman, K.} 2011.
\newblock A framework for quantitative analysis of cascades on networks.
\newblock In {\em WSDM '11: Proceedings of the fourth ACM international
  conference on Web search and data mining}. 665--674.

\bibitem[\protect\citeauthoryear{Gomez-Rodriguez, Balduzzi, and
  Sch\"{o}lkopf}{Gomez-Rodriguez et~al\mbox{.}}{2011}]{manuel11icml}
{\sc Gomez-Rodriguez, M.}, {\sc Balduzzi, D.}, {\sc and} {\sc Sch\"{o}lkopf,
  B.} 2011.
\newblock {Uncovering the Temporal Dynamics of Diffusion Networks}.
\newblock In {\em ICML '11: Proceedings of the 28th International Conference on
  Machine Learning}. 561--568.

\bibitem[\protect\citeauthoryear{Goodman}{Goodman}{1961}]{goodman61sampling}
{\sc Goodman, L.~A.} 1961.
\newblock Snowball sampling.
\newblock {\em Annals of Mathematical Statistics\/}~{\em 32,\/}~1, 148--170.

\bibitem[\protect\citeauthoryear{Goyal, Bonchi, and Lakshmanan}{Goyal
  et~al\mbox{.}}{2010}]{goyal2010learning}
{\sc Goyal, A.}, {\sc Bonchi, F.}, {\sc and} {\sc Lakshmanan, L.} 2010.
\newblock {Learning influence probabilities in social networks}.
\newblock In {\em WSDM '10: Proceedings of the Third ACM International
  Conference on Web Search and Data Mining}. ACM, 241--250.

\bibitem[\protect\citeauthoryear{Gruhl, Guha, Liben-Nowell, and Tomkins}{Gruhl
  et~al\mbox{.}}{2004}]{gruhl2004information}
{\sc Gruhl, D.}, {\sc Guha, R.}, {\sc Liben-Nowell, D.}, {\sc and} {\sc
  Tomkins, A.} 2004.
\newblock {Information diffusion through blogspace}.
\newblock In {\em WWW '04: Proceedings of the 13th international conference on
  World Wide Web}. 491--501.

\bibitem[\protect\citeauthoryear{Heckathorn}{Heckathorn}{1997}]{heckathorn97sampling}
{\sc Heckathorn, D.} 1997.
\newblock Respondent-driven sampling: A new approach to the study of hidden
  populations.
\newblock {\em Social Problems\/}~{\em 44,\/}~2, 174--199.

\bibitem[\protect\citeauthoryear{Hethcote}{Hethcote}{2000}]{hethcote00diseases}
{\sc Hethcote, H.~W.} 2000.
\newblock The mathematics of infectious diseases.
\newblock {\em SIAM Review\/}~{\em 42,\/}~4, 599--653.

\bibitem[\protect\citeauthoryear{Jansen, Yu, Greenbaum, Kluger, Krogan, Chung,
  Emili, Snyder, Greeblatt, and Gerstein}{Jansen
  et~al\mbox{.}}{2003}]{janse03linkpred}
{\sc Jansen, R.}, {\sc Yu, H.}, {\sc Greenbaum, D.}, {\sc Kluger, Y.}, {\sc
  Krogan, N.}, {\sc Chung, S.}, {\sc Emili, A.}, {\sc Snyder, M.}, {\sc
  Greeblatt, J.}, {\sc and} {\sc Gerstein, M.} 2003.
\newblock {A bayesian networks approach for predicting proteinprotein
  interactions from genomic data}.
\newblock {\em Science\/}~{\em 302,\/}~5644, 449--453.

\bibitem[\protect\citeauthoryear{Katz and Lazarsfeld}{Katz and
  Lazarsfeld}{1955}]{katz1955personal}
{\sc Katz, E.} {\sc and} {\sc Lazarsfeld, P.} 1955.
\newblock {\em {Personal influence: The part played by people in the flow of
  mass communications}}.
\newblock Free Press.

\bibitem[\protect\citeauthoryear{Kearns, Suri, and Montfort}{Kearns
  et~al\mbox{.}}{2006}]{kearns06experimental}
{\sc Kearns, M.}, {\sc Suri, S.}, {\sc and} {\sc Montfort, N.} 2006.
\newblock {An experimental study of the coloring problem on human subject
  networks}.
\newblock {\em Science\/}~{\em 313,\/}~5788, 824.

\bibitem[\protect\citeauthoryear{Kempe, Kleinberg, and Tardos}{Kempe
  et~al\mbox{.}}{2003}]{kempe03maximizing}
{\sc Kempe, D.}, {\sc Kleinberg, J.~M.}, {\sc and} {\sc Tardos, E.} 2003.
\newblock Maximizing the spread of influence through a social network.
\newblock In {\em KDD '03: Proceedings of the 9th ACM SIGKDD international
  conference on Knowledge discovery and data mining}. 137--146.

\bibitem[\protect\citeauthoryear{Khuller, Moss, and Naor}{Khuller
  et~al\mbox{.}}{1999}]{khuller1999budgeted}
{\sc Khuller, S.}, {\sc Moss, A.}, {\sc and} {\sc Naor, J.} 1999.
\newblock {The budgeted maximum coverage problem}.
\newblock {\em Information Processing Letters\/}~{\em 70,\/}~1, 39--45.

\bibitem[\protect\citeauthoryear{Knuth}{Knuth}{1968}]{knuth1968art}
{\sc Knuth, D.} 1968.
\newblock {\em {The art of computer programming}}.
\newblock Addison-Wesley.

\bibitem[\protect\citeauthoryear{Kumar, Novak, Raghavan, and Tomkins}{Kumar
  et~al\mbox{.}}{2004}]{kumar2004structure}
{\sc Kumar, R.}, {\sc Novak, J.}, {\sc Raghavan, P.}, {\sc and} {\sc Tomkins,
  A.} 2004.
\newblock {Structure and evolution of blogspace}.
\newblock {\em CACM\/}~{\em 47,\/}~12, 35--39.

\bibitem[\protect\citeauthoryear{Leskovec, Adamic, and Huberman}{Leskovec
  et~al\mbox{.}}{2006}]{jure06viral}
{\sc Leskovec, J.}, {\sc Adamic, L.~A.}, {\sc and} {\sc Huberman, B.~A.} 2006.
\newblock The dynamics of viral marketing.
\newblock In {\em EC '06: Proceedings of the 7th ACM conference on Electronic
  commerce}. 228--237.

\bibitem[\protect\citeauthoryear{Leskovec, Backstrom, and Kleinberg}{Leskovec
  et~al\mbox{.}}{2009}]{leskovec2009kdd}
{\sc Leskovec, J.}, {\sc Backstrom, L.}, {\sc and} {\sc Kleinberg, J.} 2009.
\newblock Meme-tracking and the dynamics of the news cycle.
\newblock In {\em KDD '09: Proceedings of the 15th ACM SIGKDD international
  conference on Knowledge discovery and data mining}. ACM, New York, NY, USA,
  497--506.

\bibitem[\protect\citeauthoryear{Leskovec and Faloutsos}{Leskovec and
  Faloutsos}{2007}]{leskovec2007scalable}
{\sc Leskovec, J.} {\sc and} {\sc Faloutsos, C.} 2007.
\newblock {Scalable modeling of real graphs using kronecker multiplication}.
\newblock In {\em ICML '07: Proceedings of the 24th International Conference on
  Machine Learning}. 504.

\bibitem[\protect\citeauthoryear{Leskovec, Kleinberg, and Faloutsos}{Leskovec
  et~al\mbox{.}}{2005}]{leskovec2005graphs}
{\sc Leskovec, J.}, {\sc Kleinberg, J.}, {\sc and} {\sc Faloutsos, C.} 2005.
\newblock {Graphs over time: densification laws, shrinking diameters and
  possible explanations}.
\newblock In {\em KDD '05: Proceedings of the 11th ACM SIGKDD international
  conference on Knowledge discovery in data mining}. 187.

\bibitem[\protect\citeauthoryear{Leskovec, Kleinberg, and Faloutsos}{Leskovec
  et~al\mbox{.}}{2007}]{jure07evolution}
{\sc Leskovec, J.}, {\sc Kleinberg, J.~M.}, {\sc and} {\sc Faloutsos, C.} 2007.
\newblock Graph evolution: Densification and shrinking diameters.
\newblock {\em ACM Transactions on Knowledge Discovery from Data (TKDD)\/}~{\em
  1,\/}~1, 2.

\bibitem[\protect\citeauthoryear{Leskovec, Krause, Guestrin, Faloutsos,
  VanBriesen, and Glance}{Leskovec et~al\mbox{.}}{2007}]{leskovec2007cost}
{\sc Leskovec, J.}, {\sc Krause, A.}, {\sc Guestrin, C.}, {\sc Faloutsos, C.},
  {\sc VanBriesen, J.}, {\sc and} {\sc Glance, N.} 2007.
\newblock {Cost-effective outbreak detection in networks}.
\newblock In {\em KDD '07: Proceedings of the 13th ACM SIGKDD international
  conference on Knowledge discovery and data mining}. 420--429.

\bibitem[\protect\citeauthoryear{Leskovec, Lang, Dasgupta, and
  Mahoney}{Leskovec et~al\mbox{.}}{2008}]{jure08ncp}
{\sc Leskovec, J.}, {\sc Lang, K.~J.}, {\sc Dasgupta, A.}, {\sc and} {\sc
  Mahoney, M.~W.} 2008.
\newblock Statistical properties of community structure in large social and
  information networks.
\newblock In {\em WWW '08: Proceedings of the 17th International Conference on
  World Wide Web}.

\bibitem[\protect\citeauthoryear{Leskovec, McGlohon, Faloutsos, Glance, and
  Hurst}{Leskovec et~al\mbox{.}}{2007}]{jure07cascades}
{\sc Leskovec, J.}, {\sc McGlohon, M.}, {\sc Faloutsos, C.}, {\sc Glance, N.},
  {\sc and} {\sc Hurst, M.} 2007.
\newblock Cascading behavior in large blog graphs.
\newblock In {\em SDM '07: Proceedings of the SIAM Conference on Data Mining}.

\bibitem[\protect\citeauthoryear{Leskovec, Singh, and Kleinberg}{Leskovec
  et~al\mbox{.}}{2006}]{jure06influence}
{\sc Leskovec, J.}, {\sc Singh, A.}, {\sc and} {\sc Kleinberg, J.~M.} 2006.
\newblock Patterns of influence in a recommendation network.
\newblock In {\em PAKDD '06: Proceedings of the 10th Pacific-Asia Conference on
  Knowledge Discovery and Data Mining}. 380--389.

\bibitem[\protect\citeauthoryear{Liben-Nowell and Kleinberg}{Liben-Nowell and
  Kleinberg}{2003}]{libennowell03linkpred}
{\sc Liben-Nowell, D.} {\sc and} {\sc Kleinberg, J.} 2003.
\newblock {The link prediction problem for social networks}.
\newblock In {\em CIKM '03: Proceedings of the International Conference on
  Information and Knowledge Management}. 556--559.

\bibitem[\protect\citeauthoryear{Liben-Nowell and Kleinberg}{Liben-Nowell and
  Kleinberg}{2008}]{nowell08letter}
{\sc Liben-Nowell, D.} {\sc and} {\sc Kleinberg, J.} 2008.
\newblock Tracing the flow of information on a global scale using {I}nternet
  chain-letter data.
\newblock {\em PNAS '08: Proceedings of the National Academy of
  Sciences\/}~{\em 105,\/}~12 (25 Mar.), 4633--4638.

\bibitem[\protect\citeauthoryear{Lippert, Stegle, Ghahramani, and
  Borgwardt}{Lippert et~al\mbox{.}}{2009}]{lippert2008kernel}
{\sc Lippert, C.}, {\sc Stegle, O.}, {\sc Ghahramani, Z.}, {\sc and} {\sc
  Borgwardt, K.} 2009.
\newblock {A kernel method for unsupervised structured network inference}.
\newblock In {\em AISTATS '09: Proceedings of the Artificial Intelligence and
  Statistics}.

\bibitem[\protect\citeauthoryear{Malmgren, Stouffer, Motter, and
  Amaral}{Malmgren et~al\mbox{.}}{2008}]{malgrem08poisson}
{\sc Malmgren, R.~D.}, {\sc Stouffer, D.~B.}, {\sc Motter, A.~E.}, {\sc and}
  {\sc Amaral, L.~A. A.~N.} 2008.
\newblock A poissonian explanation for heavy tails in e-mail communication.
\newblock {\em Proceedings of the National Academy of Sciences\/}~{\em
  105,\/}~47 (November), 18153--18158.

\bibitem[\protect\citeauthoryear{Meinshausen and Buehlmann}{Meinshausen and
  Buehlmann}{2006}]{meinshausen2006high}
{\sc Meinshausen, N.} {\sc and} {\sc Buehlmann, P.} 2006.
\newblock {High-dimensional graphs and variable selection with the lasso}.
\newblock {\em The Annals of Statistics\/}, 1436--1462.

\bibitem[\protect\citeauthoryear{Myers and Leskovec}{Myers and
  Leskovec}{2010}]{meyers10netinf}
{\sc Myers, S.} {\sc and} {\sc Leskovec, J.} 2010.
\newblock {On the Convexity of Latent Social Network Inference}.
\newblock In {\em NIPS '10: Advances in Neural Information Processing Systems}.

\bibitem[\protect\citeauthoryear{Nemhauser, Wolsey, and Fisher}{Nemhauser
  et~al\mbox{.}}{1978}]{nemhauser1978analysis}
{\sc Nemhauser, G.}, {\sc Wolsey, L.}, {\sc and} {\sc Fisher, M.} 1978.
\newblock {An analysis of approximations for maximizing submodular set
  functions}.
\newblock {\em Mathematical Programming\/}~{\em 14,\/}~1, 265--294.

\bibitem[\protect\citeauthoryear{Rogers}{Rogers}{1995}]{rogers95diffusion}
{\sc Rogers, E.~M.} 1995.
\newblock {\em Diffusion of Innovations\/}, Fourth ed.
\newblock Free Press, New York.

\bibitem[\protect\citeauthoryear{Romero, Meeder, and Kleinberg}{Romero
  et~al\mbox{.}}{2011}]{romero2011differences}
{\sc Romero, D.}, {\sc Meeder, B.}, {\sc and} {\sc Kleinberg, J.} 2011.
\newblock Differences in the mechanics of information diffusion across topics:
  Idioms, political hashtags, and complex contagion on twitter.
\newblock In {\em WWW '11: Proceedings of the 20th international conference on
  World wide web}. ACM, 695--704.

\bibitem[\protect\citeauthoryear{Sadikov, Medina, Leskovec, and
  Garcia-Molina}{Sadikov et~al\mbox{.}}{2011}]{sadikov11cascades}
{\sc Sadikov, S.}, {\sc Medina, M.}, {\sc Leskovec, J.}, {\sc and} {\sc
  Garcia-Molina, H.} 2011.
\newblock Correcting for missing data in information cascades.
\newblock In {\em WSDM '11: Proceedings of the ACM International Conference on
  Web Search and Data Mining}.

\bibitem[\protect\citeauthoryear{Schmidt, Niculescu-Mizil, and Murphy}{Schmidt
  et~al\mbox{.}}{2007}]{schmidt2007learning}
{\sc Schmidt, M.}, {\sc Niculescu-Mizil, A.}, {\sc and} {\sc Murphy, K.} 2007.
\newblock {Learning graphical model structure using l1-regularization paths}.
\newblock In {\em AAAI '07: Proceedings of the 21th Conference on Artificial
  Intelligence}. Vol.~22.

\bibitem[\protect\citeauthoryear{Song, Kolar, and Xing}{Song
  et~al\mbox{.}}{2009}]{song09timevarying}
{\sc Song, L.}, {\sc Kolar, M.}, {\sc and} {\sc Xing, E.} 2009.
\newblock {Time-varying dynamic bayesian networks}.
\newblock In {\em NIPS '09: Advances in Neural Information Processing Systems}.

\bibitem[\protect\citeauthoryear{Strang and Soule}{Strang and
  Soule}{1998}]{strang98diffusion}
{\sc Strang, D.} {\sc and} {\sc Soule, S.~A.} 1998.
\newblock Diffusion in organizations and social movements: From hybrid corn to
  poison pills.
\newblock {\em Annual Review of Sociology\/}~{\em 24}, 265--290.

\bibitem[\protect\citeauthoryear{Taskar, Wong, Abbeel, and Koller}{Taskar
  et~al\mbox{.}}{2003}]{taskar03linkpred}
{\sc Taskar, B.}, {\sc Wong, M.~F.}, {\sc Abbeel, P.}, {\sc and} {\sc Koller,
  D.} 2003.
\newblock {Link prediction in relational data}.
\newblock In {\em NIPS '03: Advances in Neural Information Processing Systems}.

\bibitem[\protect\citeauthoryear{Tutte}{Tutte}{1948}]{tutte48matrix}
{\sc Tutte, W.} 1948.
\newblock The disection of equilateral triangles into equilateral triangles.
\newblock {\em Proceedings Cambridge Philos. Soc.\/}~{\em 44}, 463--482.

\bibitem[\protect\citeauthoryear{Ver~Steeg, Ghosh, and Lerman}{Ver~Steeg
  et~al\mbox{.}}{2011}]{ver2011stops}
{\sc Ver~Steeg, G.}, {\sc Ghosh, R.}, {\sc and} {\sc Lerman, K.} 2011.
\newblock What stops social epidemics?
\newblock In {\em ICWSM '11: Proceedings of the 5th Int. Conf. on Weblogs and
  Social Media}.

\bibitem[\protect\citeauthoryear{Vert and Yamanishi}{Vert and
  Yamanishi}{2005}]{vert2005supervised}
{\sc Vert, J.} {\sc and} {\sc Yamanishi, Y.} 2005.
\newblock {Supervised graph inference}.
\newblock In {\em NIPS '05: Advances in Neural Information Processing Systems}.

\bibitem[\protect\citeauthoryear{Wainwright, Ravikumar, and
  Lafferty}{Wainwright et~al\mbox{.}}{2006}]{wainwright06graphical}
{\sc Wainwright, M.~J.}, {\sc Ravikumar, P.}, {\sc and} {\sc Lafferty, J.~D.}
  2006.
\newblock {High-dimensional graphical model selection using `1-regularized
  logistic regression}.
\newblock In {\em PNAS '06: Proceedings of the National Academy of Sciences}.

\bibitem[\protect\citeauthoryear{Wallinga and Teunis}{Wallinga and
  Teunis}{2004}]{wallinga04epidemic}
{\sc Wallinga, J.} {\sc and} {\sc Teunis, P.} 2004.
\newblock Different epidemic curves for severe acute respiratory syndrome
  reveal similar impacts of control measures.
\newblock {\em American Journal of Epidemiology\/}~{\em 160,\/}~6, 509--516.

\bibitem[\protect\citeauthoryear{Watts and Dodds}{Watts and
  Dodds}{2007}]{dodds07influentials}
{\sc Watts, D.~J.} {\sc and} {\sc Dodds, P.~S.} 2007.
\newblock Influentials, networks, and public opinion formation.
\newblock {\em Journal of Consumer Research\/}~{\em 34,\/}~4 (December),
  441--458.

\end{thebibliography}

\end{document}